\documentclass[12pt]{article}
\usepackage[utf8]{inputenc}

\usepackage{natbib}
\newfam\hebfam
\font\tmp=rcjhbltx at11pt \textfont\hebfam=\tmp

\edef\declfam{\ifcase\hebfam 
     0\or1\or2\or3\or4\or5\or6\or7\or8\or9\or A\or B\or C\or D\or E\or F\fi}

\mathchardef\tsadi   = "0\declfam 6C

\usepackage{appendix}
\usepackage{listings}
\usepackage{color} 
\usepackage{tikz}
\usetikzlibrary{matrix,backgrounds}
\usepackage{pgfplots}
\pgfplotsset{compat=1.16}
\usepackage{amsfonts}
\usepackage{enumitem}
\usepackage{mathtools, nccmath}
\usepackage{multicol}
\usepackage{arydshln}
\usepackage{authblk}
\usepackage{multirow}
\usepackage{graphicx}
\usepackage{amssymb}
\usepackage{cancel}
\usepackage{booktabs}
\usepackage{ulem}
\usepackage{tikz-network}
\usepackage{setspace}
\usepackage[left=1.2in,right=1.2in,top=1.2in,bottom=1.2in]{geometry}
\usepackage{verbatim}
\usepackage{graphicx}
\graphicspath{ {./figures/} }
\usepackage{subcaption}
\usepackage{amsmath}
\usepackage{lineno}
\usepackage{amsthm}
\usepackage{pgfkeys}
\newtheorem{theorem}{Theorem}

\newtheorem{lemma}{Lemma}

\newtheorem{assumption}{Assumption}
\newtheorem{proposition}{Proposition}

\usepackage{setspace}
\usepackage{caption}
\title{The Cointegrated Matrix Autoregressive Model\thanks{We thank Paolo Paruolo, Luca Barbaglia,  Alain Hecq, Tomas Del Barrio Castro, Annika Camehl, and the participants to the Third Italian Conference on Economic Statistics (2025, Naples), the Computational and Financial Econometrics Conference 2025 (London), and the 33$^{rd}$ Symposium of the Society of for Nonlinear Dynamics and Econometrics (2026, Lisbon) for stimulating comments and suggestions. Both authors acknowledge financial support from project MUR 2022C799SX - \textit{PRICE: A New Paradigm for High-Frequency Finance}, funded by the European Union - Next Generation EU, Mission 4 Component 2, CUP C53D23002680001.}}

\author[a]{Emanuele Lopetuso\thanks{Corresponding author: University of Padova, Department of Statistical Sciences, e-mail: emanuele.lopetuso@unipd.it}}
\author[a]{Massimiliano Caporin}

\affil[a]{Department of Statistical Sciences, University of Padova, Italy} 
\date{}

\begin{document}

\maketitle

\begin{abstract}
Traditional econometric analyzes represent observations as vectors despite the inherent complexity of empirical data structures. When data are organized along dual classification dimensions, a matrix representation provides a more natural and interpretable framework. Building on recent advances in matrix autoregressive (MAR) modeling, this study introduces a novel error correction representation tailored for matrix-structured data. Through comparative analysis with existing methodologies, we demonstrate two critical advancements. First, the proposed model preserves the interpretative foundations of conventional cointegration analysis, with coefficients that explicitly capture dynamics rooted in adjustment toward steady-state positions. Second, in contrast to previous formulations, our error correction framework allows for an equivalent matrix autoregressive representation, preserving the fundamental structure of the data in both specifications. This ensures that the matrix representation reflects an intrinsic characteristic of the data.
\\
\textbf{Keywords}:
 Matrix AutoRegression; Cointegration \\
\end{abstract}


\section{Introduction}

Multivariate time-series data have traditionally been treated as vectors despite the fact that the relationships among variables can give rise to more complex structures. Recent studies have focused on observations arranged in matrix form, extending the vector autoregressive (hereafter VAR) model to its matrix counterpart, commonly referred to as the matrix autoregressive (hereafter MAR) model.

Matrix-valued time series naturally arise when data are collected at the intersection of two classifications. In economics, a common example is the analysis of macroeconomic indicators across countries or regions. For example, \cite{wang2009bayesian} examines the simultaneous evolution of employment statistics in US states and industrial sectors, while \cite{zhang2024additive} investigates relationships among the main macroeconomic indicators in different countries. In such cases, it is natural to structure the data as matrices, with the economic indicators on the rows and the corresponding regional entities on the columns (or vice versa).

Another prominent application involves the analysis of bilateral trade flows, where the matrix representation effectively captures import-export relationships between countries. This framework facilitates the modeling of trade dynamics, as illustrated in \cite{chen2019modeling}.

In finance, stock market trends exhibit dependencies not only over time and across economic variables but also across different sectors, such as manufacturing and transportation. Furthermore, measures of financial connectedness are often derived from sequences of adjacency matrices that represent networks of financial assets; \cite{billio2021matrix}. These applications have spurred the development of high-dimensional methods for matrix-valued time series, including autoregressive models, \cite{billio2023bayesian}, matrix panel regression models, \cite{kapetanios2021estimation}, and matrix factor models, \cite{chen2023statistical}.

The utility of matrix-valued time series extends beyond economics. In spatio-temporal analysis, where classifications often involve geographical and temporal dimensions, the matrix framework provides a natural representation (e.g., \citet{sun2023matrix} and \citet{yu2024dynamic}). Neuroscience is another domain where matrix-valued data is prevalent. Functional magnetic resonance (fMRI) data, for example, capture interactions across different brain regions, with dependencies often visualized through matrices representing connectivity patterns; \cite{samadi2025matrix}.

The autoregressive representations are of particular interest in economics. The MAR model, as proposed in \cite{chen2021autoregressive}, is expressed as:
\begin{equation}\label{MAR}
    X_t = \Lambda X_{t-1} \Psi' + E_t,
\end{equation}
where $X_t$ is an $m \times n$ matrix, $\Lambda$ and $\Psi$ are square matrices of dimensions $m$ and $n$, respectively, and $E_t$ is an $m \times n$ white noise matrix. The vectorization of \eqref{MAR} yields the following:
\begin{equation}\label{VAR}
    \text{vec}(X_t) = (\Psi \otimes \Lambda) \text{vec}(X_{t-1}) + \text{vec}(E_t),
\end{equation}
which corresponds to the classical VAR form under specific (Kronecker-based) parameter restrictions.

The matrix representation offers two main advantages. First, it provides a significant reduction in dimensionality, decreasing the number of coefficients from $n^2m^2$ to $n^2 + m^2$. Second, it preserves the inherent structure of the data, facilitating the interpretation of relationships between rows and columns. This paper focuses on the latter advantage, examining the interpretability of MAR coefficients in the context of non-stationary and cointegrated systems.

Multivariate cointegrated systems, such as the vector error correction model (hereafter VECM), are widely valued for their economic interpretability. The VECM is represented as:
\begin{equation}\label{VECM}
    \Delta Y_t = \alpha \beta' Y_{t-1} + \sum_{i=1}^k \Gamma_i \Delta Y_{t-i} + \varepsilon_t,
\end{equation}
where $Y_t$ is a $k-$dimensional vector, $\alpha$ and $\beta$ are $k \times r$ matrices, and $r<k$ is the cointegration rank. Economically, the adjustment matrix $\alpha$ captures the behavior of agents responding to equilibrium errors by driving the economic variables back to the steady-state position defined by $\beta' Y = 0$. Specifically, $\beta$ characterizes the relations to which the variables are attracted, while $\alpha$ quantifies the direction and speed of their adjustment in response to deviations from equilibrium.

The connection between the coefficients of the VECM and the economic dynamics allows us to test numerous relevant economic hypotheses (see \citet{johansen1991estimation}). In particular, hypotheses about the cointegration matrix $\beta$ enable the assessment of the equilibrium relationships suggested in the macroeconomic literature. Furthermore, we can interpret hypotheses on the adjustment matrix $\alpha$ as assertions about adjustment mechanisms, providing insights into which variables remain unaffected by disequilibrium conditions, possibly establishing non-causality relationships (see, for example, \citet{juselius2007taking} and \citet{juselius2017real}).\footnote{A null row in the adjustment matrix $\alpha$ implies weak exogeneity for the long-run coefficients $\beta$, as defined in \cite{engle1983exogeneity} (see \citet{johansen1992testing} for a formal proof).} Moreover, testing restrictions on $\alpha$ allows the identification of the variables serving as driving trends.

This paper extends the analytical framework of cointegrated systems to matrix-valued time series. This topic has recently received attention, with a proposal put forward by \cite{li2024cointegrated} and the work of \cite{hecq2024detecting}. In their modeling strategy, the cointegrated matrix structure is obtained by imposing restrictions on the coefficients of a vector error-correction model and is driven by a bilinear equilibrium condition. We take a different perspective and derive the error-correction representation directly from an underlying matrix autoregressive (MAR) specification. This distinction is important. In our framework, the matrix structure is preserved both in the autoregressive representation in levels and in the corresponding error-correction form, so that the matrix nature of the model is intrinsic to the data-generating process, rather than being a consequence of a particular parametrization. Moreover, our model is driven by two sets of linear equilibrium conditions acting separately but jointly, over rows and over columns. Further, we show that existing cointegrated models for matrix-valued time series do not generally share this property, as their matrix form depends on whether the model is written in autoregressive or error-correction form.

The second contribution of the paper is interpretative. In our formulation, the left and right cointegration matrices admit separate economic interpretations, as they characterize long-run equilibrium relations across the rows and columns of the observed matrix process, respectively. This feature is not available in existing alternatives, where the equilibrium condition is intrinsically bilinear and therefore combines row and column effects into a single object. As a result, those models do not allow one to disentangle the equilibrium structure along the two dimensions of the data. Our approach, instead, makes it possible to study row and column steady states separately, thereby providing a more flexible and informative framework for testing economically meaningful hypotheses, in close analogy with the role played by cointegration vectors in the standard VECM literature.

We provide further contributions related to the estimation of the model. We discuss both the identification of the two cointegration ranks that characterized the cointegrated MAR model and the estimation of the cointegration matrices. In addition, we also tackle an inferential perspective and introduce approaches for testing hypotheses on the cointegration and adjustment matrices. Finally, we present an empirical analysis that examines the relationship between industrial production and inflation expectations in a selected group of European countries, demonstrating the practical application of the proposed methodology. The goodness of fit of the model is assessed by verifying the stationarity of the estimated cointegration relationships and by comparing the matrix-based estimation with the standard vector representation. Furthermore, we illustrate that employing a conventional Vector Error Correction Model does not facilitate an interpretation in terms of parallel structures. This highlights the improvement in interpretability achieved through our proposed model.

The paper is structured as follows. Section 2 proposes the model. Section 3 clarifies our contribution to the literature by comparing our model with existing formulations and highlighting the strengths of our approach. Section 4 presents an estimation algorithm and outlines procedures for testing restrictions on model coefficients and rank identification, while Section 5 evaluates the performance of these methods through Monte Carlo experiments. Section 6 provides an empirical illustration, and Section 7 concludes. Proofs of the lemmas and theorems are provided in Appendix A.

\section{Cointegrated-MAR model}

Consider the matrix autoregressive process of order one
\begin{equation}\label{MAR_uno}
    X_t = \Lambda X_{t-1}\Psi' + E_t,
\end{equation}
where $\Lambda\in\mathbb{R}^{m\times m}$, $\Psi\in\mathbb{R}^{n\times n}$, $E_t$ is a $m\times n$ white-noise process, and $X_t$ is a $m\times n$ observation of a matrix-valued time series.

\begin{assumption}\label{Ass_1}
Let $\lambda_1,\dots,\lambda_m$ and $\mu_1,\dots,\mu_n$ denote the eigenvalues of
$\Lambda$ and $\Psi$, respectively. Assume that:
\begin{enumerate}
    \item there exist integers $0<r_1<m$ and $0<r_2<n$ such that $\Lambda$ has a semisimple unit eigenvalue with algebraic multiplicity $m-r_1$ and $\Psi$ has a semisimple unit eigenvalue with algebraic multiplicity $n-r_2$;

    \item every other eigenvalue of $\Lambda$ and $\Psi$ lies strictly inside the unit disk, that is,
    \[
        |\lambda_i|<1 \quad \text{for all } \lambda_i\neq 1,
        \qquad
        |\mu_j|<1 \quad \text{for all } \mu_j\neq 1.
    \]
\end{enumerate}
\end{assumption}

\begin{theorem}\label{T01}
Let $x_t:=\operatorname{vec}(X_t)$. If Assumption \eqref{Ass_1} holds, then $x_t$ is a $I(1)$ cointegrated process, with
$(m-r_1)(n-r_2)$ common stochastic trends and cointegration rank
\[
mn-(m-r_1)(n-r_2)=nr_1+mr_2-r_1r_2.
\]
\end{theorem}

Theorem \ref{T01} shows that the vectorized process $x_t=\operatorname{vec}(X_t)$ is cointegrated. Therefore, the matrix autoregressive representation in \eqref{MAR_uno}, together with Assumption \ref{Ass_1}, gives rise to a cointegrated model with MAR representation, which we call a Cointegrated Matrix Autoregressive model (C-MAR).

We now study the possibility of a matrix error-correction representation. From \eqref{MAR_uno}, one obtains the standard vector error-correction form
\begin{equation}\label{VECM_2}
    \operatorname{vec}(\Delta X_t)
    =
    \big(\Psi\otimes\Lambda-I_{mn}\big)\operatorname{vec}(X_{t-1})
    + \operatorname{vec}(E_t).
\end{equation}

Consider the following lemma.
\begin{lemma}\label{L1}
Under Assumption \ref{Ass_1}, the matrix
\[
\Psi\otimes\Lambda-I_{mn}
\]
does not admit a single Kronecker product decomposition. More precisely, there do not exist matrices
\[
A^*\in\mathbb{R}^{n\times n},\qquad B^*\in\mathbb{R}^{m\times m}
\]
such that
\[
\Psi\otimes\Lambda-I_{mn}=A^*\otimes B^*.
\]
\end{lemma}
It follows from Lemma \ref{L1} that the impact matrix of the VECM associated with a C-MAR model cannot be written in matrix form.

To find an alternative matrix representation, note that under Assumption \ref{Ass_1},
\[
\Pi_1:=\Lambda-I_m,
\qquad
\Pi_2:=\Psi-I_n
\]
have ranks $r_1$ and $r_2$, respectively. Hence, rank factorizations of the form
\[
\Pi_1=\tau\gamma',
\qquad
\Pi_2=\varphi\theta',
\]
are admissible, with $\tau,\gamma\in\mathbb{R}^{m\times r_1}$ and $\varphi,\theta\in\mathbb{R}^{n\times r_2}$ of full column rank. It follows that
\begin{equation}\label{C-MAR}
    X_t = (I_m + \tau \gamma') X_{t-1} (I_n + \theta \varphi') + E_t.
\end{equation}

Solving for $\Delta X_t=X_t-X_{t-1}$ yields the error-correction representation
\begin{equation}\label{ECC-MAR}
    \Delta X_t = \tau \gamma' X_{t-1} + X_{t-1} \theta \varphi' + \tau \gamma' X_{t-1} \theta \varphi' + E_t.
\end{equation}
We refer to model \eqref{ECC-MAR} as the error-correction representation of the C-MAR model (hereafter, ECC-MAR). 

As emphasized in the introduction, cointegrated models allow for an interpretation in terms of adjustment toward a long-run equilibrium. In a standard VECM, the equilibrium condition is given by $\beta'x_t=0$. In fact, although $x_t$ is $I(1)$, the linear combination $\beta'x_t$ is $I(0)$. This means that the system reverts toward the position $\beta'x_t=0$ despite the variables being non-stationary.

To understand whether ECC-MAR carries the same interpretation, we study the existence of mean-reverting linear combinations associated with the ECC-MAR coefficients $\gamma$ and $\theta$.

\begin{theorem}\label{Teo_3}
Under Assumption \ref{Ass_1}, the matrix-valued processes
\[
\gamma'X_t \in \mathbb{R}^{r_1\times n},
\qquad
X_t\theta \in \mathbb{R}^{m\times r_2}
\]
in \eqref{ECC-MAR} are both $I(0)$.
\end{theorem}

It follows from Theorem \ref{Teo_3} that processes $\gamma'X_t$ and $X_t\theta$ are mean reverting. Consequently, $\gamma'X_t$ and $X_t\theta$ can be interpreted as two distinct equilibrium errors: one pertaining to the relations among the rows of $X_t$, and the other pertaining to the relations among the columns of $X_t$.

If both $\gamma'X_{t-1}=0$ and $X_{t-1}\theta=0$, then the deterministic error-correction component in \eqref{ECC-MAR} vanishes, and therefore $\mathbb{E}(\Delta X_t\mid \mathcal{F}_{t-1})=0.$ Thus, if both $\gamma'X_t=0$ and $X_t\theta=0$, there is no systematic tendency for the process to move away from its current position. In this case, the deterministic error-correction component vanishes, so the system is in equilibrium. Consequently, $\gamma$ and $\theta$ play an analogous role to the cointegration matrix $\beta$ in the standard VECM, and may be interpreted as row and column cointegration matrices, respectively.

Similarly, matrices $\tau$ and $\varphi$ play a role analogous to the adjustment matrix in a VECM. They capture the speed of adjustment toward the long-run equilibrium and indicate which variables do not respond to disequilibrium conditions.

To illustrate the adjustment mechanism implied by ECC-MAR, rewrite \eqref{ECC-MAR} as
\[
    \Delta X_t = \tau \big(\gamma' X_{t-1} + \gamma' X_{t-1} \theta \varphi'\big)
    + X_{t-1} \theta \varphi' + E_t.
\]
The first term in parentheses reflects the row disequilibrium at time $t-1$. The second component,
$\gamma' X_{t-1} \theta \varphi'$, captures the row disequilibrium induced by the adjustment mechanism acting on the columns of $X_t$. Indeed, $X_{t-1}\theta\varphi'$ represents the system's response to deviations from the column equilibrium condition $X_{t-1}\theta=0$, while $\gamma'X_{t-1}\theta\varphi'$ measures the row disequilibrium generated by that column correction. Therefore, the adjustment toward row equilibrium is driven by the composite disequilibrium term $\gamma' X_{t-1} + \gamma' X_{t-1}\theta\varphi'.$ The matrix $\tau$ determines how such deviations are corrected over time.

A symmetric argument applies to the column disequilibrium. In particular, the adjustment coefficient $\varphi$ governs the response to the composite disequilibrium term $X_{t-1}\theta + \tau\gamma'X_{t-1}\theta.$

In conclusion, the third term in \eqref{ECC-MAR}, namely $\tau\gamma'X_{t-1}\theta\varphi',$ reconciles the two parallel adjustment mechanisms: one targeting $\gamma'X_t=0$ and the other targeting $X_t\theta=0$. These two adjustment mechanisms may otherwise point in different directions, and the interaction term ensures their consistency within a unified matrix error-correction structure.

Finally, we study the relationship between ECC-MAR and its vector counterpart. Vectorizing \eqref{ECC-MAR}, we obtain
\begin{equation}\label{VECMECC-MAR}
    \operatorname{vec}(\Delta X_t)
    =
    \big(I_n \otimes \tau \gamma'
    + \varphi \theta' \otimes I_m
    + \varphi \theta' \otimes \tau \gamma'\big)\operatorname{vec}(X_{t-1})
    + \operatorname{vec}(E_t).
\end{equation}

The following theorem proves that the impact matrix of the VECM generated by the ECC-MAR representation has a reduced rank and provides a representation of the cointegration and adjustment matrices.

\begin{theorem}\label{T2}
The cointegration space of the VECM generated by the ECC-MAR representation is spanned by
\begin{equation}\label{cointegration}
    \beta = \left[\, I_n \otimes \gamma \;,\; \theta \otimes \gamma_\perp \,\right].
\end{equation}
A corresponding adjustment matrix is
\begin{equation}\label{adjustementt}
    \alpha =
    \left[\,
    I_n \otimes \tau + (\varphi\theta')\otimes(\bar{\gamma}+\tau)
    \;,\;
    \varphi\otimes\bar{\gamma}_\perp
    \,\right],
\end{equation}
where $\bar a = a(a'a)^{-1}$ for any matrix $a$ with full-column-rank.
\end{theorem}

\subsection{Cointegrated-MAR of order p}

The natural extension of the ECC-MAR representation in equation \eqref{ECC-MAR} is
\begin{equation}\label{ECC-MAROrderP}
    \Delta X_t=\tau \gamma'X_{t-1}+X_{t-1}\theta \varphi'+\tau\gamma'X_{t-1}\theta \varphi'+\sum_{i=1}^{p-1}\Gamma_{1,i }\Delta X_{t-i}\Gamma'_{2,i}+E_t.
\end{equation}

The correspondence between the ECC-MAR above and the C-MAR($p$) is not straightforward as for the case where the autoregressive form is of order 1, and requires assumptions on the relations among the MAR coefficients. To illustrate, consider the MAR($p$) model
\begin{equation}\label{C-MARp}
    X_t=\sum_{i=1}^p \Lambda_i X_{t-i} \Psi'_i+E_t
\end{equation}
and its alternative representation
\begin{equation*}
    X_t=\sum_{i=1}^p \Lambda_i X_{t-1} \Psi'_i-\sum_{j=1}^{p-1}\sum_{i=j+1}^{p}\Lambda_i \Delta X_{t-j} \Psi'_i+E_t
\end{equation*}
Vectorizing, one obtains
\begin{equation*}
    \text{vec}(X_t)=\Bigg[\sum_{i=1}^p \Psi_i \otimes \Lambda_i \Bigg]\text{vec}(X_{t-1}) -\sum_{j=1}^{p-1}\Bigg[\sum_{i=j+1}^p \Psi_i \otimes \Lambda_i \Bigg] \text{vec}(\Delta X_{t-j}) +\text{vec}(E_t)
\end{equation*}
To obtain the ECC-MAR representation in \eqref{ECC-MAROrderP}, it is necessary that $\sum_{i=j}^p \Psi_i \otimes \Lambda_i=A_j \otimes B_j$ $\forall j=\{1,2,\dots,p\}$. In other words, it is necessary that the sum of these Kronecker products is itself a Kronecker product. 

This is not generally the case, as the rank condition illustrated in Lemma \ref{L1} is not ensured for the sum of Kronecker products. The following Theorem illustrates a condition ensuring that the sum of Kronecker products is itself a Kronecker product. We first introduce two Lemmas.
\begin{lemma}\label{L2}
    If $A_1,\dots,A_n$ are all proportional to the same matrix,
i.e. $A_i = \lambda_i A_1$ for some scalars $\lambda_i$, then $C \otimes D=\sum_{i}^n A_i \otimes B_i$. Similarly, if $B_1,\dots,B_n$ are all proportional to the same matrix,
i.e., $B_i = \delta_i B_1$, then $C \otimes D=\sum_{i}^n A_i \otimes B_i$.
\end{lemma}

To ensure the duality between MAR and ECC-MAR representations, the condition that $\sum_{i=j}^p \Psi_i \otimes \Lambda_i$ is itself a Kronecker product is not sufficient. In fact, it is required that the resulting matrices can be decomposed in such a way to obtain the terms governing the adjustment mechanism, namely, it is required that $A \otimes B=\sum_{i=1}^p \Psi_i \otimes \Lambda_i$ have eigenvalues lower and equal to 1 in absolute value, so that they can be decomposed as $B=I+\Pi_1$ and $A=I+\Pi_2$, where $\Pi_{1,2}$ are rank deficient. 
\begin{lemma}\label{L3}
Assume that $A_1,\dots,A_n$ are all proportional to the same matrix, and also
$B_1,\dots,B_n$, so that $A_i=\lambda_i A_1$ and $B_i=\delta_i B_1$ for some scalars
$\lambda_i,\delta_i$. Let $s:=\sum_{i=1}^n \lambda_i\delta_i$ and assume $s\neq 0$.
If $|\text{eig}(A_1)|\le 1$ and $|\text{eig}(B_1)|\le |s|^{-1}$, then $|\text{eig}(C)|\le 1$ and
$|\text{eig}(D)|\le 1$, where $C\otimes D=\sum_{i=1}^n A_i\otimes B_i$.
\end{lemma}

The stability conditions hold in the usual way, by analyzing the eigenvalues of the companion matrix.

From Lemmas \ref{L2} and \ref{L3}, the following Theorem immediately results: 

\begin{theorem}\label{T3}
     If $\Lambda_i=\lambda_i\Lambda_j$  and $\Psi_i=\delta_j \Psi_j$ $\forall i,j=1:p$, and if $\text{eig} \vert (\Lambda_1) \vert \leq 1$ and $\vert \text{eig}(\Psi_1) \vert \leq (\sum \lambda_i \delta_i)^{-1}$, then the C-MAR($p$) model in equation \eqref{C-MARp} has an ECC-MAR representation as in equation \eqref{ECC-MAROrderP}
\end{theorem}

\section{Alternative Matrix Representation for Cointegrated Systems.}

Concepts related to cointegration and reduced-rank modeling have already been incorporated into the literature on matrix-valued time series. For example, \cite{xiao2022reduced} introduce the Reduced Rank Matrix Autoregressive model (RRMAR), which extends the MAR model in \eqref{MAR} by imposing reduced-rank restrictions on the autoregressive coefficient matrices. However, the first contribution to explicitly introduce cointegration into matrix-valued autoregressive models is \cite{li2024cointegrated}; see also the extension in \cite{hecq2024detecting}. In its simplest form, their model is
\begin{equation}\label{MECM}
    \Delta X_t = \Pi_1 X_{t-1}\Pi_2' + E_t
    = \tau\gamma' X_{t-1}\theta\varphi' + E_t,
\end{equation}
where $\Pi_1=\tau\gamma'$ and $\Pi_2=\varphi\theta'$ are reduced-rank matrices of ranks $r_1$ and $r_2$, respectively, with full-column-ranks $\tau,\gamma\in\mathbb{R}^{m\times r_1}$ and $\varphi,\theta\in\mathbb{R}^{n\times r_2}$.

The corresponding vectorized representation is a standard VECM with coefficient matrix constrained to have Kronecker-product form:
\begin{equation}\label{Vec_MECM}
    \text{vec}(\Delta X_t)
    =
    (\varphi\theta' \otimes \tau\gamma')\text{vec}(X_{t-1})+\text{vec}(E_t)
    =
    (\varphi\otimes\tau)(\theta'\otimes\gamma')\text{vec}(X_{t-1})+\text{vec}(E_t).
\end{equation}

To ensure that the vectorized process $x_t=\operatorname{vec}(X_t)$ is $I(1)$ and cointegrated, the following Johansen-type condition is imposed on the characteristic polynomial of the vectorized system.
\begin{assumption}\label{ASS_2}
Let
\begin{equation*}
     A(z)
    =
    (1-z)I_{mn}
    -z(\varphi\theta' \otimes \tau\gamma').
\end{equation*}
If $|A(z)|=0$, then either $|z|>1$ or $z=1$.
\end{assumption}

This construction differs fundamentally from the C-MAR proposed in the present paper for two main reasons. First, the model in \eqref{MECM} is postulated directly in error-correction form and is not derived from a MAR representation in levels. Second, its equilibrium structure is intrinsically bilinear and does not disentangle the equilibrium relations of the rows and columns.

To illustrate the first point, vectorizing \eqref{MECM} yields the VECM representation in \eqref{Vec_MECM}, where $\varphi\otimes\tau$ and $\theta'\otimes\gamma'$ are the adjustment and cointegration matrices, respectively. However, \eqref{MECM} does not imply a MAR representation in levels of the form \eqref{MAR}. Indeed, the MECM can be written in VAR form as
\begin{equation*}
    \text{vec}(X_t) = \big(\varphi\theta' \otimes \tau\gamma' + I_{mn}\big)\text{vec}(X_{t-1}) + \text{vec}(E_t).
\end{equation*}
Consider the following lemma.
\begin{lemma}\label{L4}
Given Assumption \ref{ASS_2}, the matrix
\[
\varphi\theta' \otimes \tau\gamma' + I_{mn}
\]
resulting from model \eqref{MECM} does not admit a single Kronecker product decomposition. More precisely, there do not exist matrices
\[
A^*\in\mathbb{R}^{n\times n},
\qquad
B^*\in\mathbb{R}^{m\times m}
\]
such that
\[
\varphi\theta' \otimes \tau\gamma' + I_{mn} = A^*\otimes B^*.
\]
\end{lemma}

Lemma \ref{L4} shows that the VAR representation associated with \eqref{MECM} does not admit a MAR representation in levels. Therefore, the term cointegrated MAR would be inappropriate. For this reason, we refer to \eqref{MECM} as a Matrix Error Correction Model (hereafter MECM), in order to distinguish it from the C-MAR developed in the previous section, which instead originates from a MAR.

The absence of a MAR representation is not merely a mathematical detail; it has important consequences for the interpretation of the equilibrium structure in the MECM. In this model, the stationary long-run component is the bilinear term $\gamma'X_t\theta$. Hence, the equilibrium condition is given by $\gamma'X_t\theta = 0,$, which in general cannot be decomposed into two separate one-sided equilibrium relations.

To further illustrate the point, consider the following theorem:
\begin{theorem}\label{Teo_4}
Under Assumption \ref{ASS_2}, the matrix-valued processes
\[
\gamma'X_t \in \mathbb{R}^{r_1\times n},
\qquad
X_t\theta \in \mathbb{R}^{m\times r_2}
\]
in \eqref{MECM} are both $I(1)$. 
\end{theorem}
Theorem \ref{Teo_4} further clarifies why the equilibrium structure of the MECM is fundamentally different from that of the C-MAR. Since the one-sided processes $\gamma'X_t$ and $X_t\theta$ are themselves $I(1)$ rather than $I(0)$, they cannot be interpreted as equilibrium errors or mean-reverting deviations from a steady state. In a cointegrated system, only stationary combinations can play this role, because an equilibrium relation must remain bounded over time and pull the system back toward a long-run position. Here, instead, the only stationary object is the bilinear combination $\gamma'X_t\theta$. Therefore, the long-run equilibrium in the MECM is intrinsically joint, involving the row and column structures simultaneously, and cannot be decomposed into two separate equilibrium relations.

In light of the previous discussion, the C-MAR proposed in this paper should be considered conceptually distinct from the matrix cointegration models already available in the literature. Although both the MECM and the C-MAR are designed to model cointegrated matrix-valued time series, they differ in both structure and interpretation. The MECM is formulated directly in error-correction form and delivers an intrinsically bilinear equilibrium condition. By contrast, the C-MAR originates from a genuine MAR representation and gives rise to two separate equilibrium errors, which admit an independent interpretation in terms of row- and column long-run relations.

\section{Estimation}

We estimate the ECC--MAR model by maximum likelihood and implement the estimation
through an alternating algorithm, in the spirit of the iterative likelihood-based
procedures proposed by \cite{li2024cointegrated}.

We begin from the general ECC--MAR$(p)$ in equation \eqref{ECC-MAROrderP}. To obtain a likelihood-based estimator, we impose the additional assumption that the innovation matrix is Gaussian with separable covariance as in \cite{chen2021autoregressive}:
\begin{equation}
E_t \stackrel{i.i.d.}{\sim} MN_{m,n}(0,\Sigma_r,\Sigma_c),
\label{eq:matrix_normal_ass_general}
\end{equation}
where $\Sigma_r \in \mathbb{R}^{m\times m}$ and
$\Sigma_c \in \mathbb{R}^{n\times n}$ are positive definite. Equivalently, $\mathrm{vec}(E_t)\sim N\!\left(0,\Sigma_c\otimes\Sigma_r\right).$ Let
\begin{equation}
\mathcal{M}_t(\vartheta)
:=
\tau \gamma' X_{t-1}
+
X_{t-1}\theta\phi'
+
\tau\gamma'X_{t-1}\theta\phi'
+
\sum_{i=1}^{p-1}\Gamma_{1,i}\Delta X_{t-i}\Gamma_{2,i}',
\label{eq:conditional_mean_general}
\end{equation}
where $\vartheta
=
\big(
\tau,\gamma,\phi,\theta,\Gamma_{1,1},\Gamma_{2,1},\dots,\Gamma_{1,p-1},\Gamma_{2,p-1}
\big),$ and define the structural residual
\begin{equation}
\mathcal{U}_t(\vartheta)
:=
\Delta X_t-\mathcal{M}_t(\vartheta).
\label{eq:structural_residual_general}
\end{equation}
Conditional on the initial values and under \eqref{eq:matrix_normal_ass_general},
we have $\mathcal{U}_t(\vartheta)=E_t \sim MN_{m,n}(0,\Sigma_r,\Sigma_c).$

It follows that the conditional likelihood of the sample is
\begin{align}
L(\vartheta,\Sigma_r,\Sigma_c)
=
(2\pi)^{-T_0mn/2}
|\Sigma_r|^{-T_0n/2}
|\Sigma_c|^{-T_0m/2}
\exp\!\left\{
-\frac{1}{2}
\sum_{t=p}^{T}
\mathrm{tr}\!\left(
\Sigma_r^{-1}\mathcal{U}_t(\vartheta)\Sigma_c^{-1}\mathcal{U}_t(\vartheta)'
\right)
\right\},
\label{eq:sample_likelihood_general}
\end{align}
where $T_0 := T-p+1$. Taking logarithms yields the Gaussian log-likelihood
\begin{align}
\ell(\vartheta,\Sigma_r,\Sigma_c)
&=
-\frac{T_0mn}{2}\log(2\pi)
-\frac{T_0n}{2}\log|\Sigma_r|
-\frac{T_0m}{2}\log|\Sigma_c|
\nonumber\\
&\quad
-\frac{1}{2}
\sum_{t=p}^{T}
\mathrm{tr}\!\left(
\Sigma_r^{-1}\mathcal{U}_t(\vartheta)\Sigma_c^{-1}\mathcal{U}_t(\vartheta)'
\right).
\label{eq:loglik_matrix_general}
\end{align}

To estimate \eqref{ECC-MAROrderP}, we employ an alternating maximum likelihood algorithm.
The idea is to update the coefficients of the row structure while holding fixed the
coefficients of the column structure, and then to proceed symmetrically in the opposite
direction. 

Note that post-multiplying \eqref{ECC-MAROrderP} by $\varphi_\perp$ eliminates the
right-hand error-correction terms and yields
\begin{equation}
\Delta X_t\varphi_\perp
=
\tau \gamma'X_{t-1}\varphi_\perp
+
\sum_{i=1}^{p-1}\Gamma_{1,i}\Delta X_{t-i}\Gamma_{2,i}'\varphi_\perp
+
E_t\varphi_\perp.
\label{eq:left_projected_raw_text}
\end{equation}
The projected errors satisfy $E_t\varphi_\perp \sim MN\!\left(0,\Sigma_r,\varphi_\perp'\Sigma_c\varphi_\perp\right)$. Hence, defining $C_\varphi:=\varphi_\perp'\Sigma_c\varphi_\perp$ and right-whitened variables $\widetilde Y_t^{L}:=\Delta X_t\varphi_\perp C_\varphi^{-1/2},$ $\widetilde X_t^{L}:=X_{t-1}\varphi_\perp C_\varphi^{-1/2},$ $
\widetilde Z_{i,t}^{L}:=\Delta X_{t-i}\Gamma_{2,i}'\varphi_\perp C_\varphi^{-1/2},$
equation \eqref{eq:left_projected_raw_text} becomes
\begin{equation}\label{Auxiliary_1}
\widetilde Y_t^{L}
=
\tau\gamma'\widetilde X_t^{L}
+
\sum_{i=1}^{p-1}\Gamma_{1,i}\widetilde Z_{i,t}^{L}
+
\widetilde E_t^{L},
\end{equation}
with $\widetilde E_t^{L}\sim MN(0,\Sigma_r,I)$.

Similarly, pre-multiplying \eqref{ECC-MAROrderP} by $\tau_\perp'$ eliminates
the left-hand error-correction terms and gives
\begin{equation}
\tau_\perp'\Delta X_t
=
\tau_\perp'X_{t-1}\theta\varphi'
+
\sum_{i=1}^{p-1}\tau_\perp'\Gamma_{1,i}\Delta X_{t-i}\Gamma_{2,i}'
+
\tau_\perp'E_t.
\label{eq:right_projected_raw_text}
\end{equation}
Again, the projected errors satisfy $\tau_\perp'E_t \sim MN\!\left(0,\tau_\perp'\Sigma_r\tau_\perp,\Sigma_c\right).$ Therefore, defining $C_\tau:=\tau_\perp'\Sigma_r\tau_\perp$ and
$\widetilde Y_t^{R}:=\Delta X_t'\tau_\perp C_\tau^{-1/2},$ $\widetilde X_t^{R}:=X_{t-1}'\tau_\perp C_\tau^{-1/2},$ $\widetilde Z_{i,t}^{R}:=\Delta X_{t-i}'\Gamma_{1,i}'\tau_\perp C_\tau^{-1/2},$ we obtain
\begin{equation}\label{Auxiliary_2}
   \widetilde Y_t^{R}
=
\varphi\theta'\widetilde X_t^{R}
+
\sum_{i=1}^{p-1}\Gamma_{2,i}\widetilde Z_{i,t}^{R}
+
\widetilde E_t^{R}, 
\end{equation}
with $\widetilde E_t^{R}\sim MN(0,\Sigma_c,I)$.

Equations \eqref{Auxiliary_1} and \eqref{Auxiliary_2} can then be estimated using standard
Gaussian reduced-rank regression methods; see \citet{anderson1951estimating} and
\citet{johansen1995likelihood}. After whitening, the auxiliary matrix equations can be
read column by column as standard multivariate reduced-rank regressions with common
coefficient matrices. The corresponding likelihood therefore decomposes as a sum over
all time-column pairs $(t,j)$, so that the estimation can be based on pooled cross-products
of the transformed variables. Appendix B illustrates the Johansen estimation procedure step by step in the present setting.

The factorization of the long-run component into $(\tau,\gamma)$ and $(\varphi,\theta)$ is not unique, as in standard reduced-rank representations, so the economically meaningful objects are the associated row and column
cointegration spaces, with specific matrix representatives selected by normalization. Likewise, under the separable Gaussian assumption, the covariance factors $\Sigma_r$ and $\Sigma_c$ are identified only up to a common scale factor, since only the Kronecker product $\Sigma_c\otimes\Sigma_r$ is uniquely determined.

The preceding results motivate the following alternating maximum likelihood procedure.

\begin{enumerate}
\item \textbf{Initialization.}
Choose initial values for the column-side parameters
$\Sigma_c^{(0)}, \varphi^{(0)},\theta^{(0)},$ $\Gamma_{2,1}^{(0)},\dots,\Gamma_{2,p-1}^{(0)}$,
and compute the corresponding orthogonal complement $\varphi_\perp^{(0)}$.
As a practical initialization device, one may use estimates obtained from a standard
vector error-correction model fitted to a representative column of $X_t$. This serves only
as a warm start for the iterative procedure.

\item \textbf{Update the row structure.}
Given
$\varphi_\perp^{(k)},\theta^{(k)},\Gamma_{2,1}^{(k)},\dots,\Gamma_{2,p-1}^{(k)},\Sigma_c^{(k)}$,
construct the transformed variables in \eqref{Auxiliary_1} and estimate the row-side
parameters by Gaussian reduced-rank regression $
\big(\tau^{(k+1)},\gamma^{(k+1)},\Gamma_{1,1}^{(k+1)},\dots,\Gamma_{1,p-1}^{(k+1)},\Sigma_r^{(k+1)}\big).
$
Then compute the corresponding orthogonal complement $\tau_\perp^{(k+1)}$.

\item \textbf{Update the column structure.}
Given
$\tau_\perp^{(k+1)},\Gamma_{1,1}^{(k+1)},\dots,\Gamma_{1,p-1}^{(k+1)},\Sigma_r^{(k+1)}$,
construct the transformed variables in \eqref{Auxiliary_2} and estimate the column-side
parameters by Gaussian reduced-rank regression,
$
\big(\varphi^{(k+1)},\theta^{(k+1)},\Gamma_{2,1}^{(k+1)},\dots,\Gamma_{2,p-1}^{(k+1)},\Sigma_c^{(k+1)}\big).
$
Then compute the corresponding orthogonal complement $\varphi_\perp^{(k+1)}$.

\item \textbf{Evaluate the system likelihood.}
Compute the system's Gaussian log-likelihood in \eqref{eq:loglik_matrix_general}.

\item \textbf{Stopping rule.}
Repeat Steps 2--4 until $
\left|\ell^{(k+1)}-\ell^{(k)}\right|<\varepsilon,
$ for a prescribed tolerance level $\varepsilon>0$. As a practical safeguard, if the system
log-likelihood decreases at iteration $k+1$, the last update can be rejected and the
algorithm stopped at the previous iterate.
\end{enumerate}

Two identification remarks are in order. First, the cointegration
and adjustment matrices are identified only up to the usual rank-preserving
normalizations, so inference is formulated in terms of the associated cointegration
spaces. Second, under the separable Gaussian assumption
$\mathrm{vec}(E_t)\sim N(0,\Sigma_c\otimes\Sigma_r)$, the covariance factors
$\Sigma_r$ and $\Sigma_c$ are identified only up to a common scale normalization.

\subsection{Determination of the Cointegration Ranks}

The estimation procedure requires prior knowledge of the cointegration ranks
$r_1$ and $r_2$. By Theorem~\ref{T01}, these ranks are linked to the cointegration
rank $r$ of the vectorized system through
\begin{equation}
r = n r_1 + m r_2 - r_1 r_2,
\label{eq:rank_relation}
\end{equation}
where $m$ and $n$ are known. Since $r_1$ and $r_2$ are ranks, they must be
non-negative integers. The first step is therefore to estimate the overall cointegration rank $r$ of the
vectorized model $\mathrm{vec}(X_t)$ using the standard Johansen tests, such as the trace test or the maximum eigenvalue test. The asymptotic
properties of these tests are well established; see Chapter~11 of
\citet{johansen1995likelihood}, in particular Theorem~11.1.

In many cases, solving \eqref{eq:rank_relation} over the set of admissible integer values
yields a unique pair $(r_1,r_2)$. For example, if $m=4$, $n=3$, and $\hat r=8$, then
\eqref{eq:rank_relation} implies the unique admissible solution $(r_1,r_2)=(2,1)$.
In this case, the values of $(r_1,r_2)$ are completely determined by $\hat r$,
so the only possible source of error is an incorrect estimate of the vectorized
cointegration rank $r$.

When more than one admissible pair satisfies \eqref{eq:rank_relation}, the ambiguity
can be resolved using the stationarity implications of the ECC--MAR model. This is
formalized in the following proposition.

\begin{proposition}\label{P_rankpairs}
Let $\{X_t\}$ be generated by an ECC-MAR process with left and right
cointegration ranks $(r_1,r_2)$ and the corresponding cointegration spaces
$\mathcal C_\gamma\subseteq\mathbb R^{m}$ and $\mathcal C_\theta\subseteq\mathbb R^{n}$,
with $\dim(\mathcal C_\gamma)=r_1$ and $\dim(\mathcal C_\theta)=r_2$.
Let $(r_1^\ast,r_2^\ast)$ be an admissible pair that satisfies the same rank identity
$r = n r_1 + m r_2 - r_1 r_2 = n r_1^\ast + m r_2^\ast - r_1^\ast r_2^\ast$.

If $(r_1^\ast,r_2^\ast)\neq(r_1,r_2)$, then at least one of the following holds:
\begin{enumerate}
\item for any full column rank $\gamma^\ast\in\mathbb R^{m\times r_1^\ast}$,
not all components of $\gamma^{\ast\prime}X_t$ are $I(0)$;
\item for any full column rank $\theta^\ast\in\mathbb R^{n\times r_2^\ast}$,
not all components of $X_t\theta^\ast$ are $I(0)$.
\end{enumerate}
\end{proposition}

Proposition~\ref{P_rankpairs} applies only when more than one admissible pair remains.
In that case, each candidate model is estimated and one checks whether the transformed
processes $\hat\gamma'X_t$ and $X_t\hat\theta$ are stationary in all components. The
correct pair $(r_1,r_2)$ is identified as the one for which both sets of transformed
series are stationary, while any candidate that overstates one of the two ranks must
generate at least one non-stationary component. Stationarity can be assessed by standard
unit root tests, such as the Dickey--Fuller test; see \citet{dickey1979distribution}.

In summary, the determination of $(r_1,r_2)$ proceeds in two steps. First, estimate
the vectorized rank $r$ and enumerate all admissible integer pairs that
satisfy \eqref{eq:rank_relation}. Second, if this yields more than one admissible pair, estimate
the competing specifications and select the one for which $\hat\gamma'X_t$ and
$X_t\hat\theta$ are stationary in all components.

\subsection{Testing Hypotheses on Cointegration and Adjustment Matrices}

From an economic point of view, the cointegration matrices $\gamma$ and $\theta$
describe the long-run equilibrium structure of the row and column spaces, respectively,
while the adjustment matrices $\tau$ and $\varphi$ describe the speed with which
deviations from these equilibrium relations are corrected. This interpretation of the
ECC--MAR coefficients makes it possible to associate parameter restrictions with
economically meaningful hypotheses.

A first class of economically relevant hypotheses concerns uniform restrictions on the
cointegration matrices:
\begin{equation}
\gamma_c = H_\gamma \widetilde{\gamma},
\qquad
\theta_c = H_\theta \widetilde{\theta},
\label{eq:test_uniform_cointegration}
\end{equation}
where $H_\gamma$ and $H_\theta$ are design matrices. These restrictions are imposed on
the entire cointegration space and are useful, for example, to test whether a variable is
long-run excludable, through a zero-row restriction, or whether a long-run homogeneity
relation holds across all equilibrium relations. Economically, these tests assess whether
a given variable contributes to the long-run structure at all or whether certain
transformed variables, such as spreads, ratios, or real quantities, provide a more
meaningful description of equilibrium behavior.

A second class of hypotheses concerns whether a known vector belongs to the
cointegration space, namely, whether
\begin{equation}
\gamma_c = (g,\widetilde{\gamma}),
\qquad
\theta_c = (t,\widetilde{\theta}),
\label{eq:test_known_theta}
\end{equation}
for known vectors $g$ and $t$. These restrictions are useful when economic theory
suggests a specific equilibrium relation, for example, a spread, a unit-elasticity relation,
or the stationarity of a single variable. In this case, the test asks whether the proposed
vector belongs to the stationarity space spanned by the estimated cointegration relations.

A third class of hypotheses concerns restrictions on the adjustment matrices:
\begin{equation}
\tau_c = H_\tau \widetilde{\tau},
\qquad
\varphi_c = H_\varphi \widetilde{\varphi}.
\label{eq:test_adjustment}
\end{equation}
Of particular interest is the case of zero-row restrictions, which corresponds to
weak exogeneity. If a row of $\tau$ or $\varphi$ is zero, the corresponding
variable does not respond to disequilibria in the long run. Economically, such variables
can be interpreted as drivers of the common stochastic trends, whereas variables with
non-zero adjustment coefficients are those that bear the burden of restoring equilibrium.
Hence, tests on the adjustment matrices help distinguish between the variables that
push the system and those that pull it back toward its long-run path.

These three classes of restrictions cover the hypotheses that are most relevant in
empirical applications of the ECC--MAR model. Many additional restrictions can be
handled. For a detailed treatment of such
cases in the standard C-VAR model, see Chapters~10 and~11 of
\citet{juselius2006cointegrated}.

The ECC--MAR model admits a natural likelihood-based testing framework through the
auxiliary regressions \eqref{Auxiliary_1} and \eqref{Auxiliary_2}. As shown above,
after whitening these auxiliary systems are standard Gaussian reduced-rank regressions,
so that inference can be conducted by likelihood ratio tests, exactly as in the classical
cointegrated VAR model. The only difference is that, in the present matrix setting,
the relevant sample moment matrices are pooled across the transformed column-wise
observations, exactly as in the estimation step. Under the null, the corresponding
likelihood ratio statistics are asymptotically chi-square, with degrees of freedom
determined by the number of imposed restrictions. Appendix~C reports the explicit
construction of each test statistic.

\section{Simulations}

We conducted a Monte Carlo study to evaluate three aspects of the proposed methodology: (i) finite-sample estimation accuracy, (ii) rank selection, and (iii) empirical performance of the restriction tests.

Data are generated from the ECC-MAR model in \eqref{ECC-MAR}. The matrices $\tau$ and $\varphi$ are constructed with the first $m-r_1$ and $n-r_2$ rows set to zero, respectively, while the remaining entries are drawn from independent standard normal distributions. The matrices $\gamma$ and $\theta$ are generated independently of standard normal distributions. This construction ensures $\mathrm{rank}(\tau\gamma')=r_1$ and $\mathrm{rank}(\varphi\theta')=r_2$, thus inducing $r_1$ row-wise and $r_2$ column-wise cointegration relations. To obtain a $I(1)$ process, we retain only parameter draws satisfying the stability conditions on the non-unit eigenvalues of $I_{r_1}+\gamma'\tau$ and $I_{r_2}+\theta'\phi$. Innovations are generated as i.i.d. matrix-valued white noise with $E[E_t]=0$ and $\mathrm{Var}(\mathrm{vec}(E_t))=I_{m\times n}$. From $X_0=0$, the process is generated recursively from \eqref{ECC-MAR}. A burn-in of 100 observations is used throughout.

We vary the sample size $T$, the matrix dimensions $(m,n)$, and the cointegration ranks $(r_1,r_2)$. Following \cite{li2024cointegrated}, we consider $(m,n)\in\{(4,3),(6,5),(8,7)\}$. For $(m,n)=(4,3)$ we set $(r_1,r_2)\in\{(1,1),(2,2),(3,2)\}$, while for the other two dimensions we consider $(r_1,r_2)\in\{(1,1),(2,2),(3,3)\}$.

\subsection{Finite-Sample Performance of the ECC-MAR Estimator}

We compare the finite-sample accuracy of the proposed ECC-MAR estimator with that of a standard CVAR estimated on the vectorized process $\mathrm{vec}(X_t)$. In this subsection, the ranks $r_1$, $r_2$, and $r$ are assumed known and fixed at their true values, so the analysis focuses exclusively on the estimation of the long-run structure.

Since individual coefficient matrices are not uniquely identified, we assess performance through the cointegration space. Let $\beta_T$ denote the true cointegration matrix implied by the DGP, constructed from the true $\gamma$ and $\theta$ according to \eqref{cointegration}. Let $\hat\beta_{\mathrm{ECC}}$ and $\hat\beta_{\mathrm{CVAR}}$ denote the corresponding estimators obtained from the ECC-MAR and CVAR models. For each estimator, we consider the projection matrix $P_\beta=\beta(\beta'\beta)^{-1}\beta'.$ For each Monte Carlo replication, the accuracy of the estimation is measured by the spectral norm distance $\|\hat P_{\beta}-P_T\|$, where $P_T$ is the projector associated with the true cointegration space.

\begin{figure}[!ht]
\centering
\resizebox{\textwidth}{0.42\textheight}{%
\begin{minipage}{\textwidth}

\begin{subfigure}{0.42\textwidth}
\centering
\includegraphics[width=\linewidth]{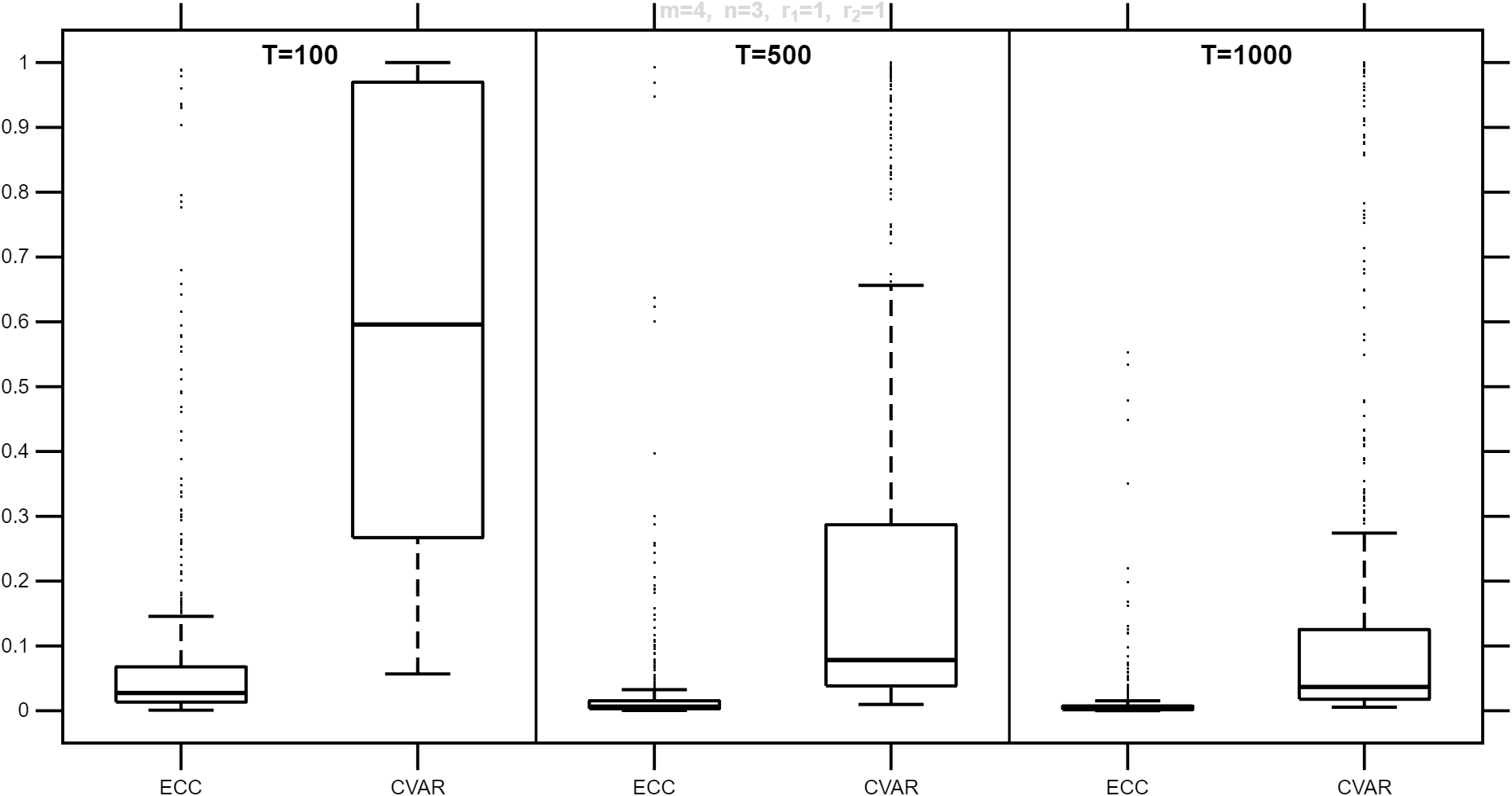}
\caption{$m=4,\ n=3,\ r_1=1,\ r_2=1$}
\end{subfigure}\hfill
\begin{subfigure}{0.42\textwidth}
\centering
\includegraphics[width=\linewidth]{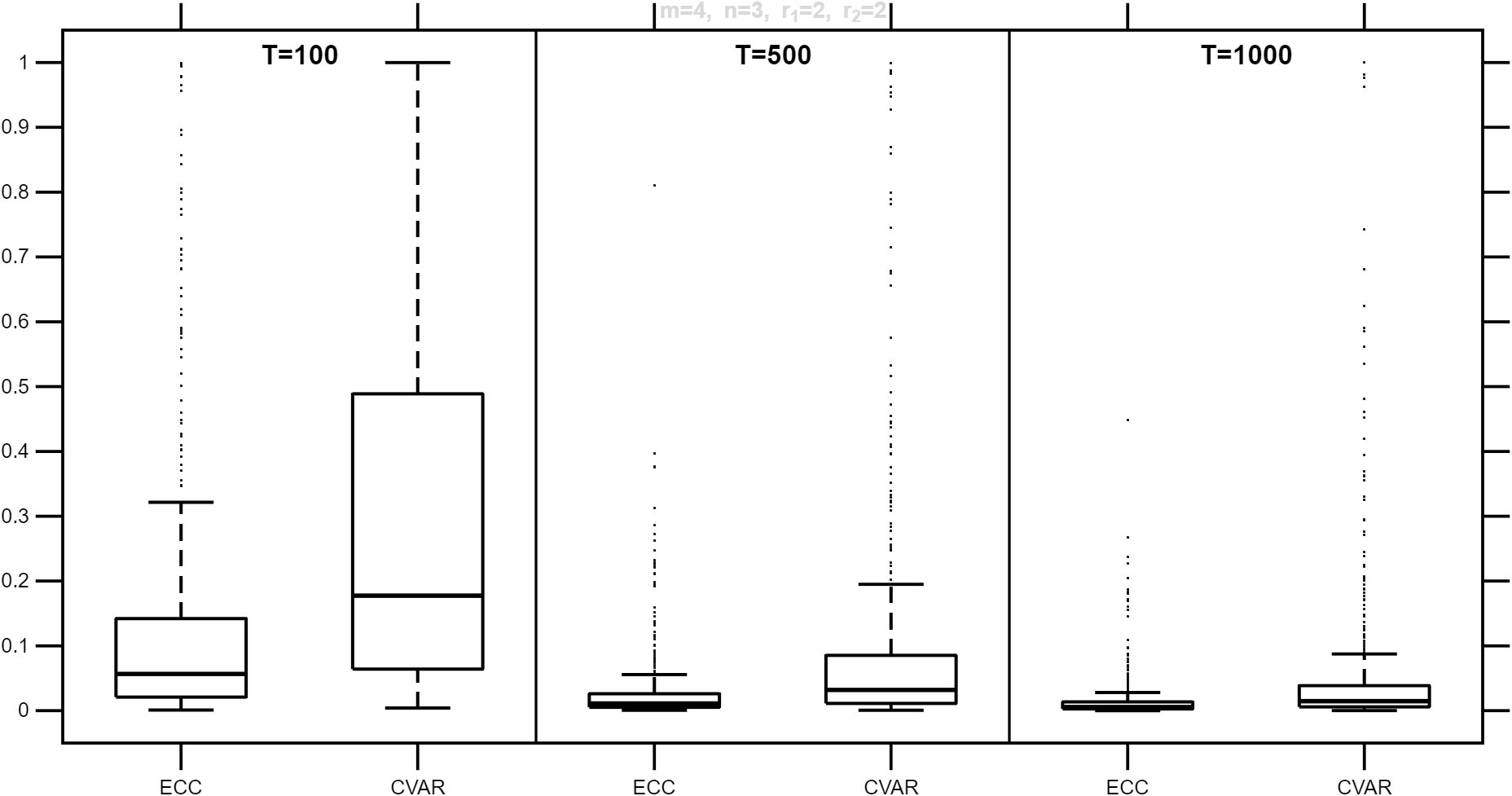}
\caption{$m=4,\ n=3,\ r_1=2,\ r_2=2$}
\end{subfigure}

\vspace{0.15cm}

\begin{subfigure}{0.42\textwidth}
\centering
\includegraphics[width=\linewidth]{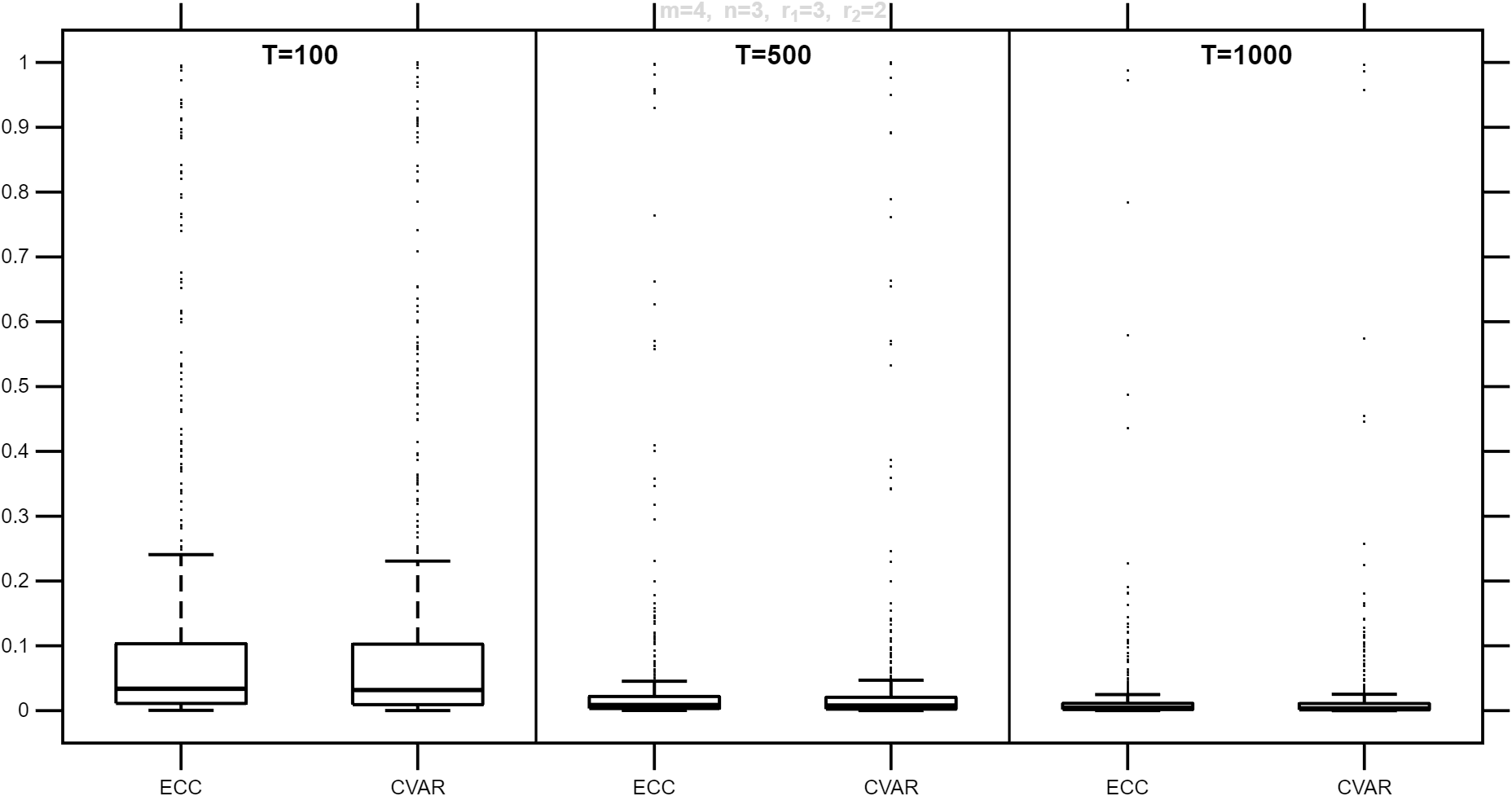}
\caption{$m=4,\ n=3,\ r_1=3,\ r_2=2$}
\end{subfigure}\hfill
\begin{subfigure}{0.42\textwidth}
\centering
\includegraphics[width=\linewidth]{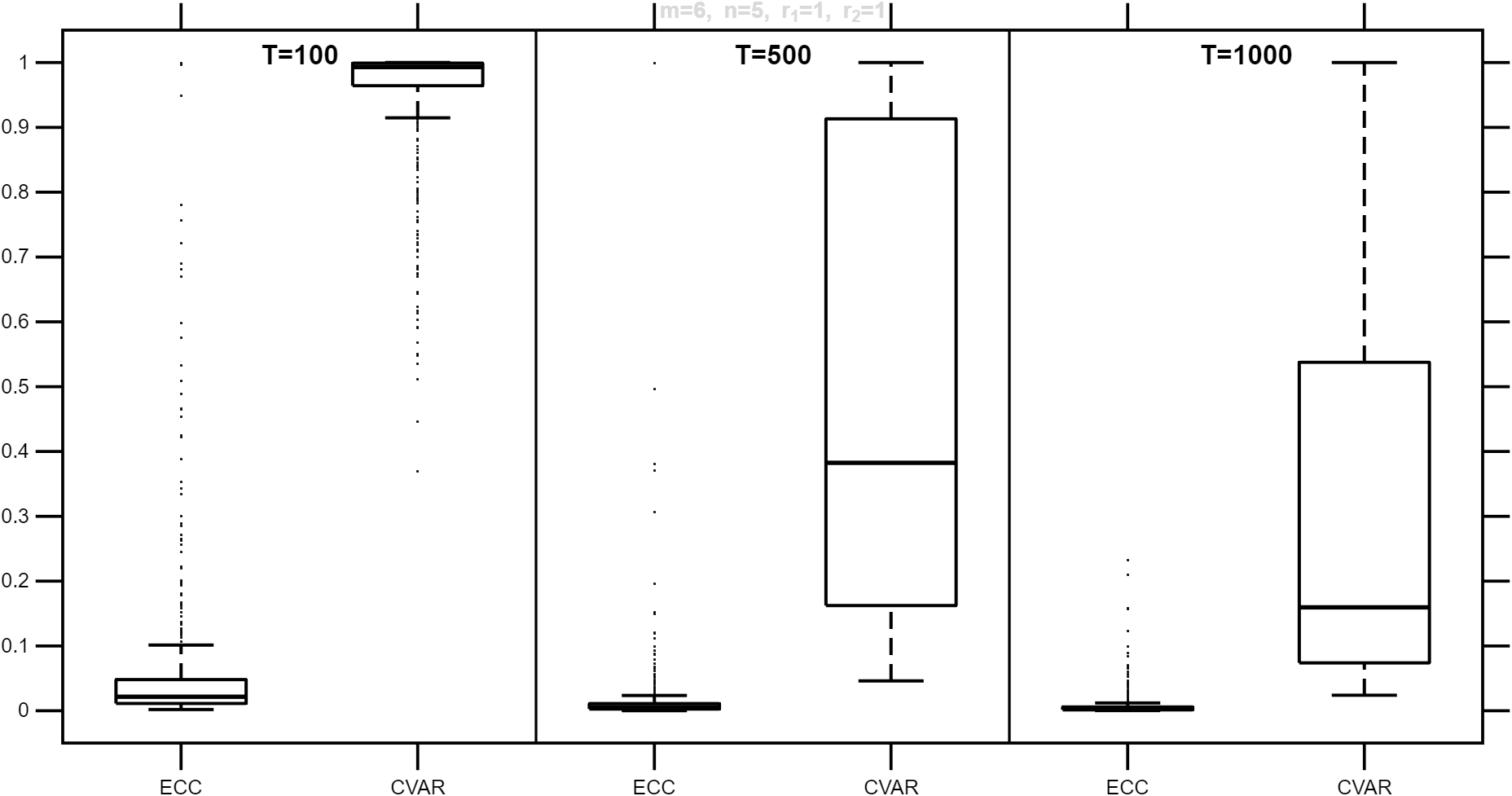}
\caption{$m=6,\ n=5,\ r_1=1,\ r_2=1$}
\end{subfigure}

\vspace{0.15cm}

\begin{subfigure}{0.42\textwidth}
\centering
\includegraphics[width=\linewidth]{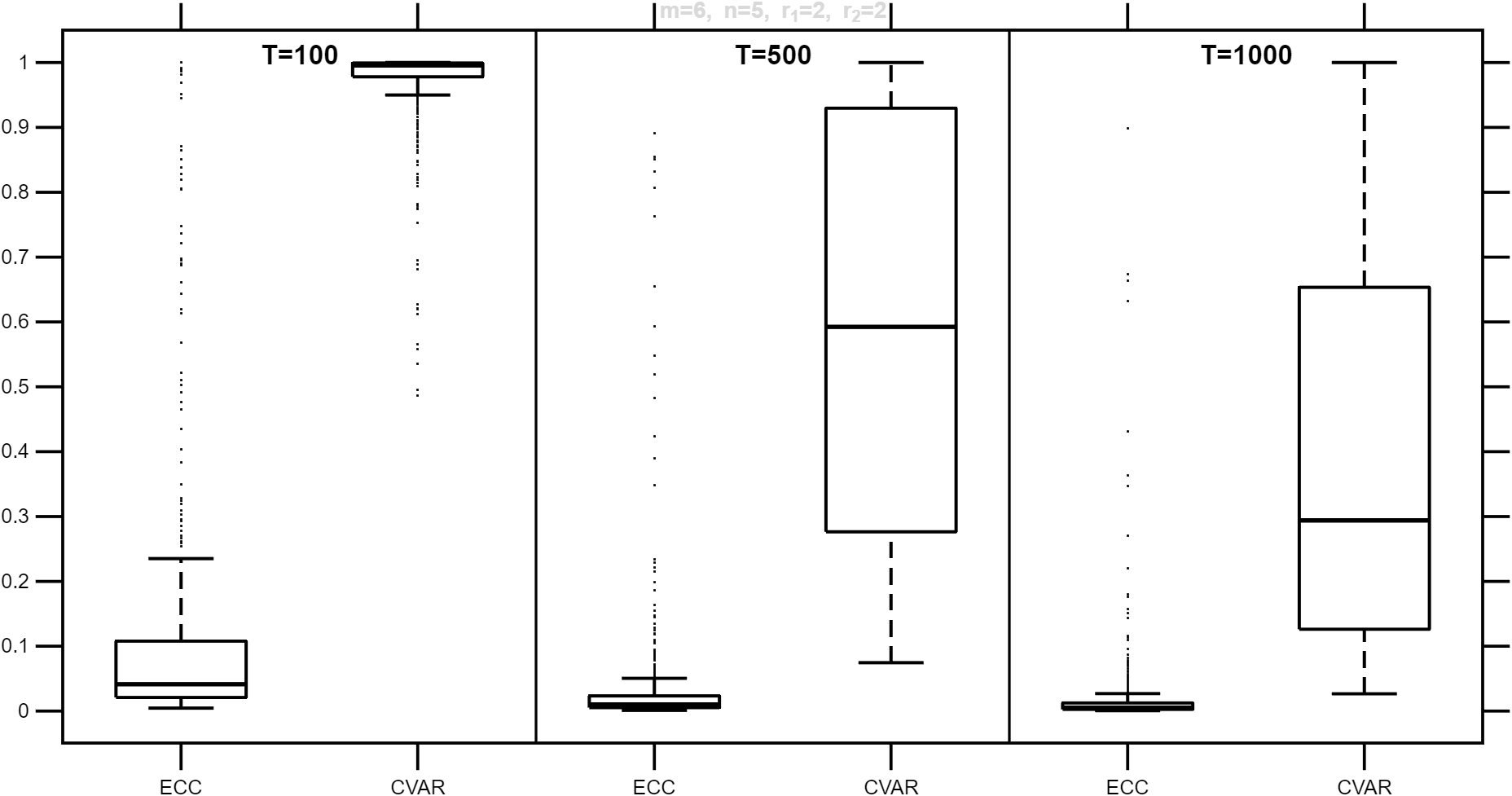}
\caption{$m=6,\ n=5,\ r_1=2,\ r_2=2$}
\end{subfigure}\hfill
\begin{subfigure}{0.42\textwidth}
\centering
\includegraphics[width=\linewidth]{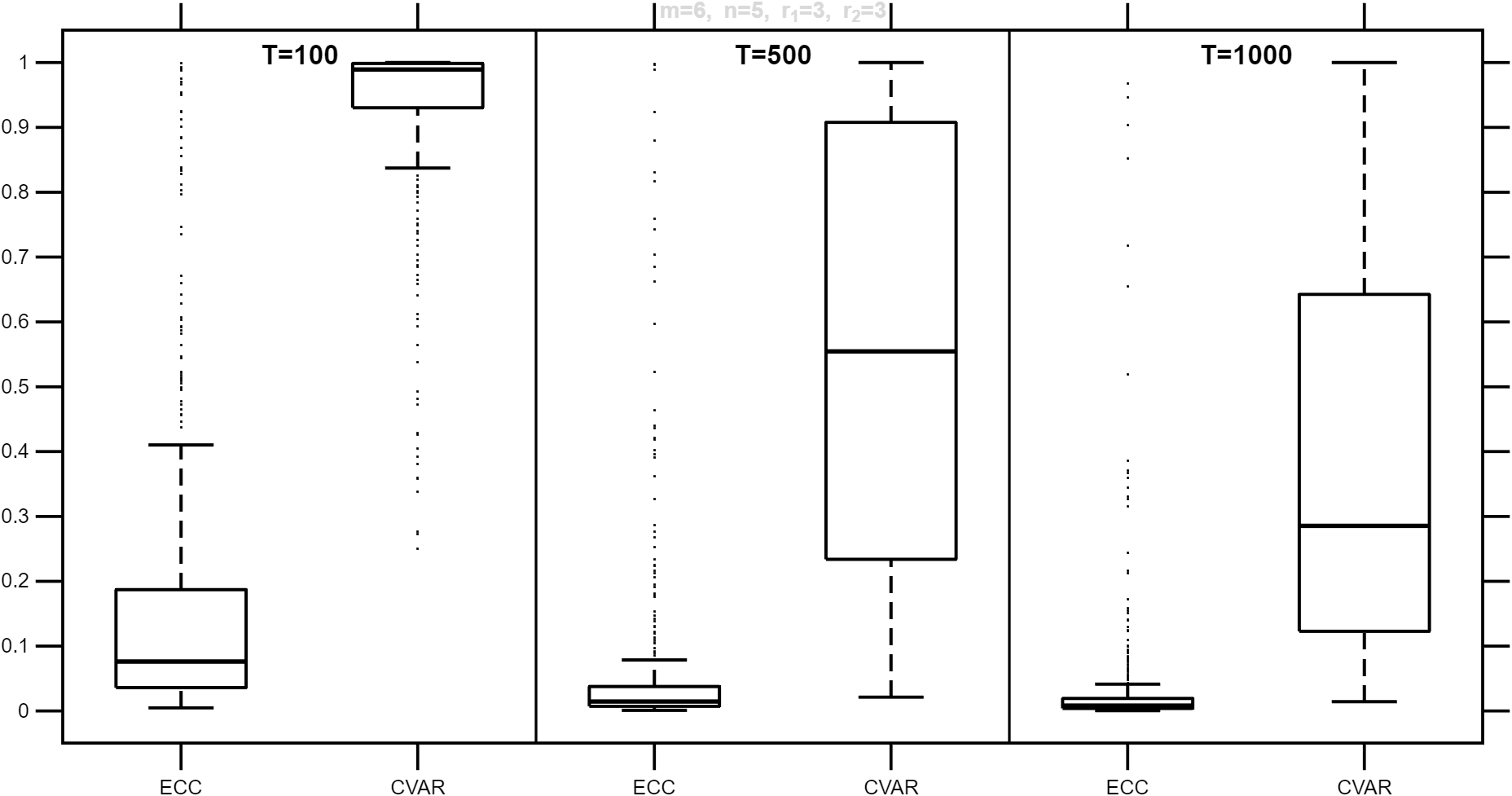}
\caption{$m=6,\ n=5,\ r_1=3,\ r_2=3$}
\end{subfigure}

\vspace{0.15cm}

\begin{subfigure}{0.42\textwidth}
\centering
\includegraphics[width=\linewidth]{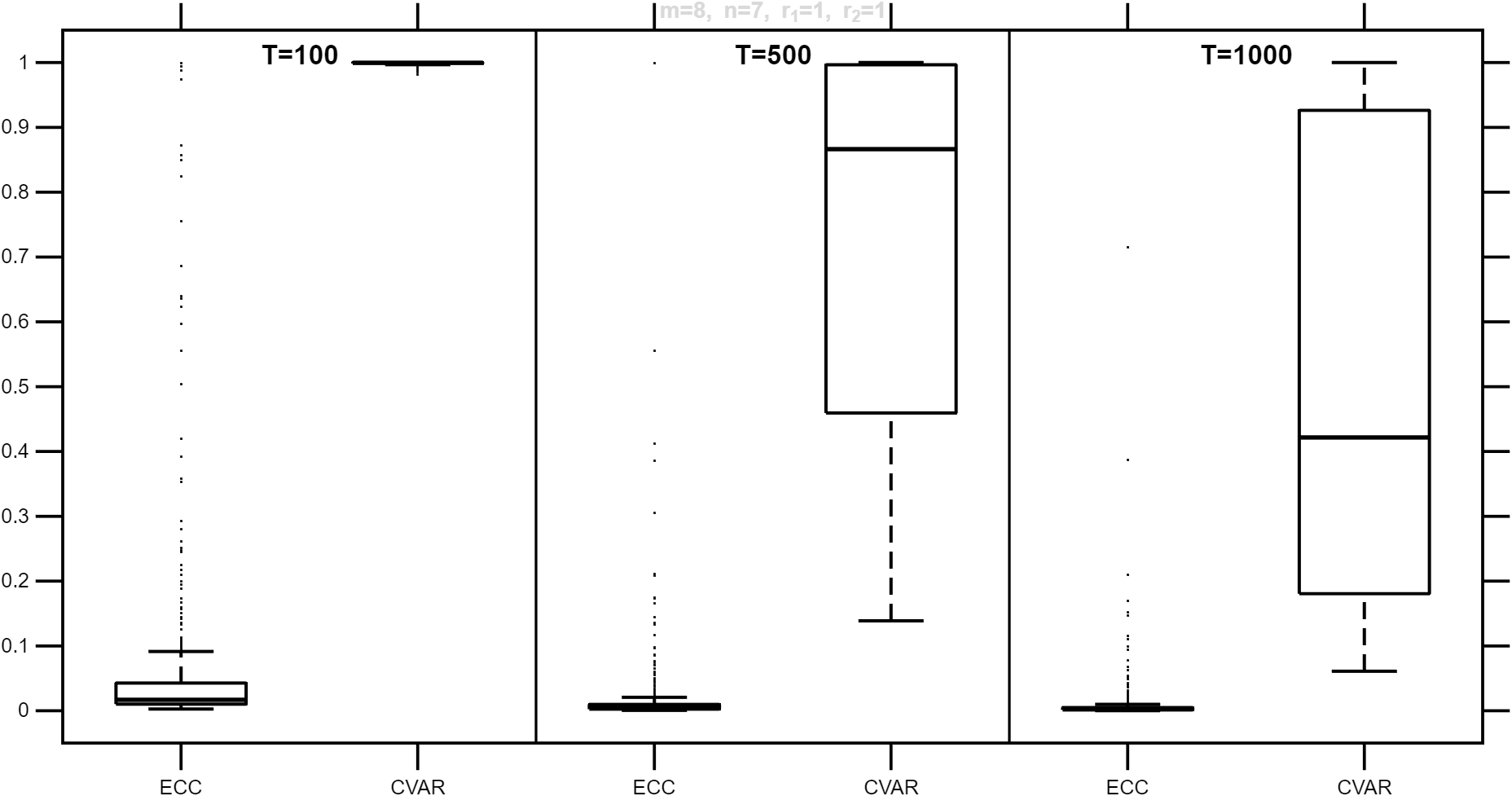}
\caption{$m=8,\ n=7,\ r_1=1,\ r_2=1$}
\end{subfigure}\hfill
\begin{subfigure}{0.42\textwidth}
\centering
\includegraphics[width=\linewidth]{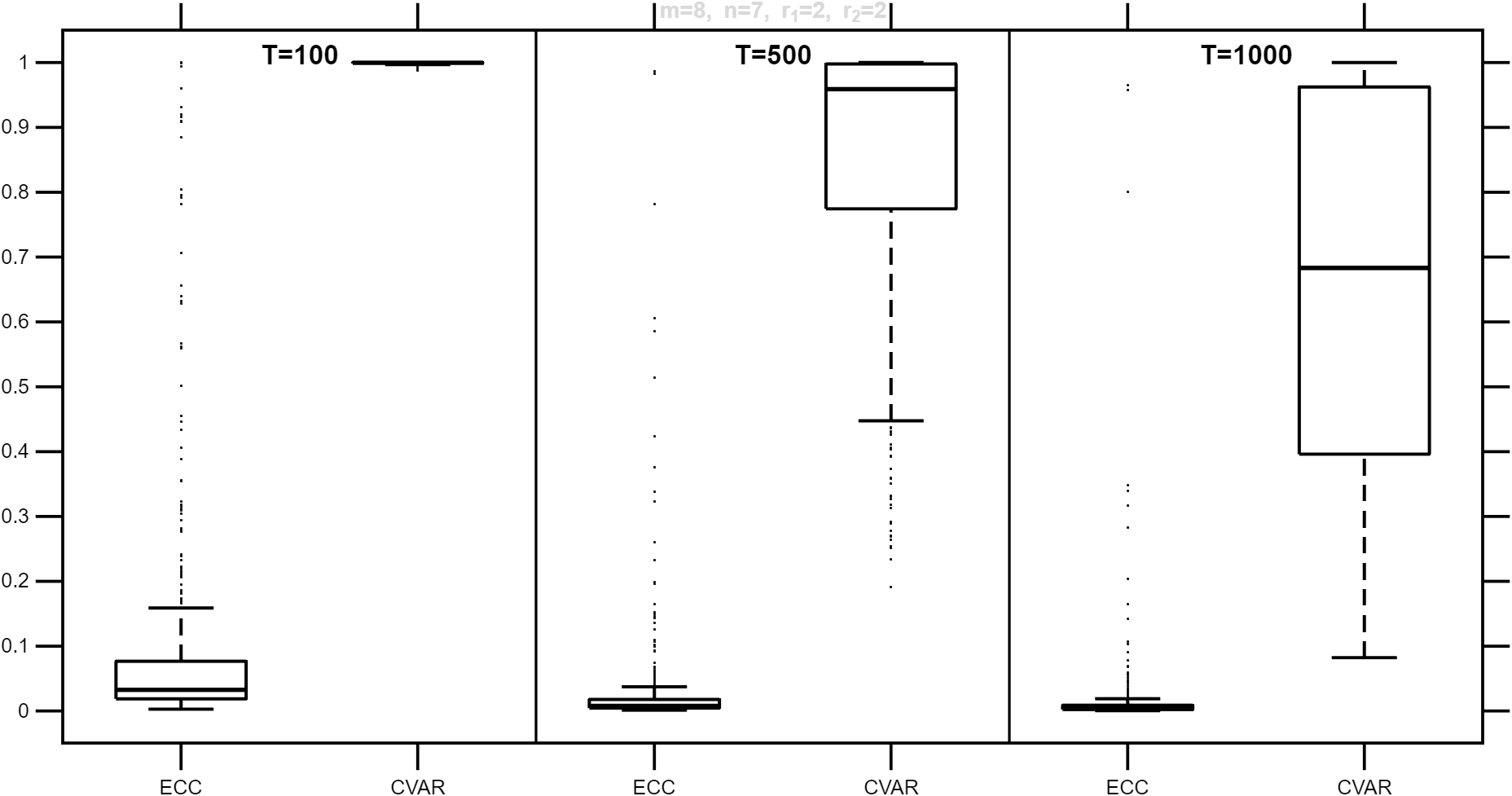}
\caption{$m=8,\ n=7,\ r_1=2,\ r_2=2$}
\end{subfigure}

\vspace{0.2cm}

\makebox[\textwidth][c]{%
\begin{subfigure}{0.5\textwidth}
\centering
\includegraphics[width=\linewidth]{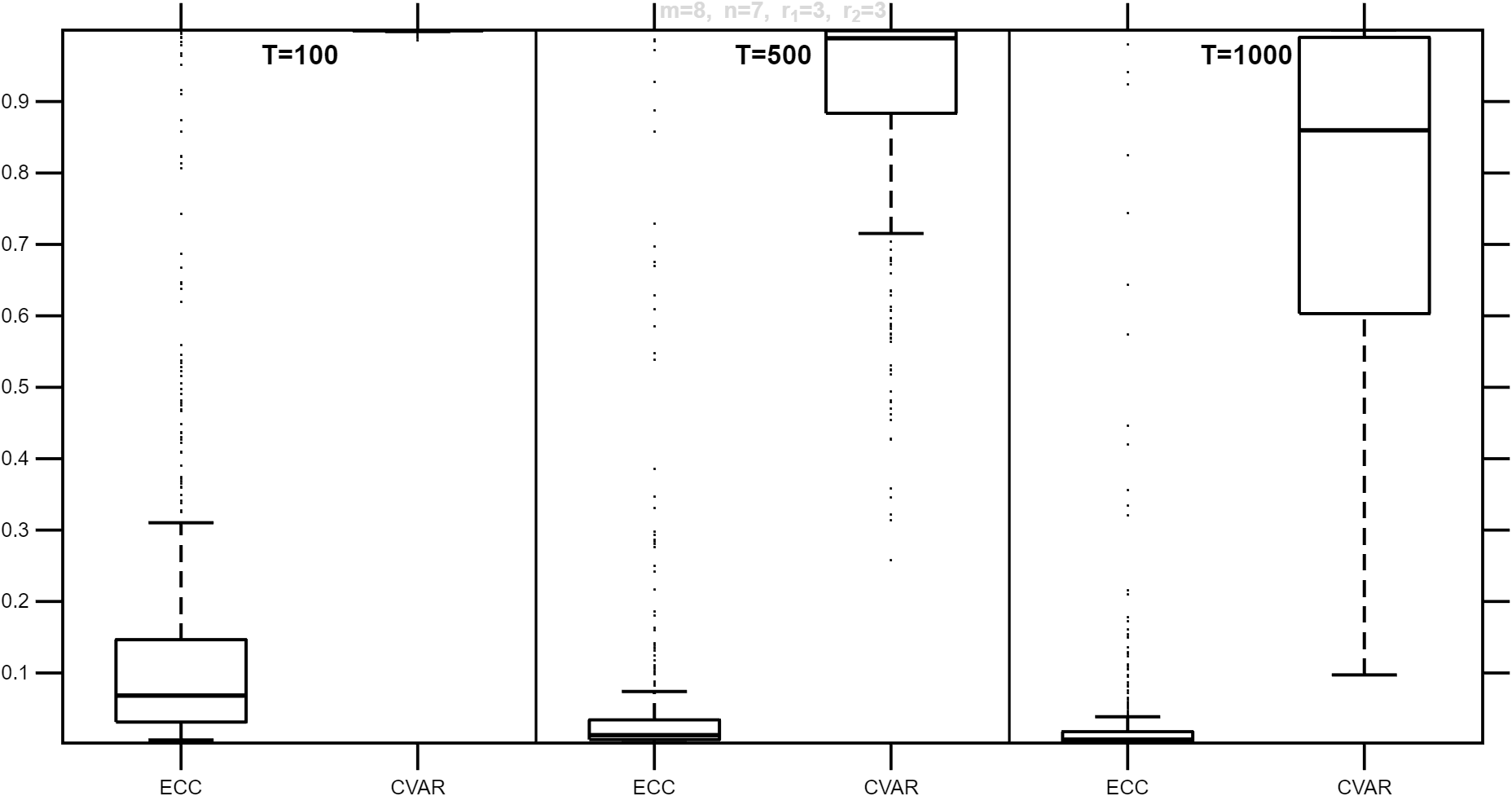}
\caption{$m=8,\ n=7,\ r_1=3,\ r_2=3$}
\end{subfigure}}

\end{minipage}
}

\caption{Boxplots of the distance between the estimated and true cointegration spaces over 500 Monte Carlo replications. For each configuration of matrix dimensions and cointegration ranks, results are reported for three sample sizes: $T=100$ (first area), $T=500$ (second area), and $T=1000$ (third area). Within each area, two boxplots are shown, corresponding to the ECC-MAR estimator (left) and the CVAR estimator (right).}
\label{fig:mc_boxplots_5x2}
\end{figure}

Figure~\ref{fig:mc_boxplots_5x2} summarizes the results. As expected, both estimators improve with the sample size, with a clear reduction in dispersion from $T=100$ to $T=1000$. In general, the ECC-MAR estimator consistently outperforms the unrestricted CVAR in recovering the cointegration space. The difference between the two estimators is larger when the cointegration ranks are low, and it tends to shrink as $r_1$ and $r_2$ increase. Equivalently, the gain from exploiting the matrix structure is more pronounced when the number of common stochastic trends, $mn-r=(m-r_1)(n-r_2)$, is large, whereas it becomes smaller when $r$ approaches $mn$.

This pattern is clearly visible within the same matrix dimension. For example, in the $(m,n)=(4,3)$ designs, the gap between ECC-MAR and CVAR is evident for $(r_1,r_2)=(1,1)$, whereas the performance of the two estimators are comparable for $(r_1,r_2)=(3,2)$. A similar pattern emerges in the other configurations. For instance, for $(m,n)=(8,7)$, the advantage of ECC-MAR is more pronounced when $(r_1,r_2)=(1,1)$ than when $(r_1,r_2)=(3,3)$.

\subsection{Identification of the Cointegration Ranks}

The first step is to determine the total cointegration rank $r$ of the vectorized process $\mathrm{vec}(X_t)$. This can be done using standard CVAR procedures, such as the Johansen trace test. Since our focus is on the identification of the structural ranks $(r_1,r_2)$ in the ECC-MAR model, we do not further investigate the determination of $r$. Moreover, misspecification of $r$ affects both the ECC-MAR and the unrestricted CVAR benchmark in a similar way. Consequently, in this subsection we treat $r$ as known.

Given $r$, the problem reduces to identifying the admissible pair $(r_1,r_2)$ consistent with \eqref{eq:rank_relation}. In four of the nine designs considered in the simulation study, this mapping is one-to-one, so the structural ranks are uniquely determined and no additional step is required. This happens for $(m,n,r)=(4,3,6)$, $(4,3,11)$, $(6,5,10)$, and $(8,7,14)$, corresponding respectively to $(r_1,r_2)=(1,1),(3,2),(1,1)$ and $(1,1)$. However, in the remaining five cases, the same total rank is compatible with two different pairs $(r_1,r_2)$.

Following Section~4.1, for each admissible pair $(r_1,r_2)$ we estimate the ECC-MAR model under the corresponding rank restrictions and obtain $\hat\gamma$ and $\hat\theta$. We then test the stationarity of the implied cointegrating components by applying Augmented Dickey--Fuller tests to the estimated series $\{\hat\gamma'X_t\}$ and $\{X_t\hat\theta\}$. The admissible pair is selected when these transformed processes are found to be stationary for one candidate but not for the other.

\begin{table}[!ht]
\centering

\renewcommand{\arraystretch}{1.2}
\begin{tabular}{lcccc}
\hline
$(m,n,r)$ & Outcome & $T=100$ & $T=500$ & $T=1000$ \\
\hline
$(4,3,10)$
& $(r_1,r_2)=(2,2)$ & 49.2\% & 69.2\% & 73.6\%  \\
& $(r_1,r_2)=(3,1)$ & 2.8\% & 2\% & 0.6\% \\
& Und. 1: Both & 14.2\% & 20\% & 19.8\% \\
& Und. 2: None & 33.8\% & 8.8\% & 6\%\\
\hline
$(6,5,18)$
& $(r_1,r_2)=(2,2)$ & 58.2\% & 83.6\%  & 88.2\%  \\
& $(r_1,r_2)=(3,1)$ & 0.4\% & 0.6\% & 0.4\% \\
& Und. 1: Both & 1.4\% & 2.4\% & 3\% \\
& Und. 2: None & 40\% & 13.4\% & 8.4\%  \\
\hline
$(6,5,24)$
& $(r_1,r_2)=(3,3)$  & 55.4\% & 79.6\% & 83.6\% \\
& $(r_1,r_2)=(4,2)$ & 2.2\% & 1.4\% & 1\% \\
& Und. 1: Both & 5\% & 8.2\% & 9\% \\
& Und. 2: None & 37.4\% & 10.8\% & 6.4\% \\
\hline
$(8,7,26)$
& $(r_1,r_2)=(2,2)$  & 57.8\% & 87.4\% & 91.8\% \\
& $(r_1,r_2)=(3,1)$ & 0\% & 0.2\% & 0.6\% \\
& Und. 1: Both & 0\% & 0.6\% & 0.6\% \\
& Und. 2: None & 42.2\% & 11.8\% & 7\%  \\
\hline
$(8,7,36)$
& $(r_1,r_2)=(3,3)$  & 56.4\% & 80.6\% & 87.6\% \\
& $(r_1,r_2)=(4,2)$ & 1.4\% & 1.4\% & 1.4\% \\
& Und. 1: Both & 0.6\% & 3.8\% & 3.2\% \\
& Und. 2: None & 41.6\% & 14.2\% & 7.8\% \\
\hline
\end{tabular}
\caption{Empirical frequencies (in percent) of rank-identification outcomes based on ADF stationarity checks over 500 Monte Carlo replications. For each configuration $(m,n,r)$ and sample size $T$, the table reports the percentage of times the procedure selects one of the admissible rank pairs $(r_1,r_2)$. The label \emph{Und.\,1} denotes an undefined outcome in which both admissible rank pairs satisfy the stationarity conditions, while \emph{Und.\,2} denotes an undefined outcome in which neither admissible rank pair satisfies the stationarity conditions.}
\label{tab:rank_identification_percentages}
\end{table}

Table~\ref{tab:rank_identification_percentages} reports the results for the five ambiguous configurations. In all cases, the correct rank pair is the outcome most frequently selected, and its frequency increases markedly with the sample size. By contrast, the incorrect admissible pair is only chosen rarely, often with frequencies close to zero even at $T=100$.

The main difficulty is the occurrence of undefined outcomes. Two cases may arise. In \emph{Und.\,1}, both admissible pairs pass the stationarity checks; in \emph{Und.\,2}, neither pair does. The second type is dominant in small samples, but decreases sharply as $T$ increases, reflecting the low power of ADF tests in short samples. For example, in configuration $(8,7,36)$ its frequency drops from 41.6\% with $T=100$ to 7.8\% with $T=1000$.

The behavior of \emph{Und.\,1} is more subtle. Its frequency tends to increase with the sample size, most visible in the case $(4,3,10)$, where it increases from 14.2\% to 19.8\%. This can occur when the model is estimated under a rank larger than the true one. In that case, the additional estimated cointegrating vector cannot belong exactly to the true cointegration space, but as the sample grows, it may become nearly collinear with it. Equivalently, the estimated cointegrating space may become increasingly close to the true one while preserving the full column rank. As a result, the extra linear combination can appear progressively more stationary, increasing the probability that the ADF test rejects the unit-root null.

We examine this mechanism in the configuration $(m,n,r)=(4,3,10)$, where the increase in \emph{Und.\,1} is strongest. Let $\gamma$ be the true cointegration matrix, $\gamma_\perp$ its orthogonal complement, and $\gamma_T$ the misspecified estimate obtained under incorrect rank $r_1=3$. Define $Q_{\gamma_\perp}=\mathrm{orth}(\gamma_\perp),$ $
Q_{\gamma,T}=\mathrm{orth}(\gamma_T),$ and $M_T = Q_{\gamma,T}'Q_{\gamma_\perp},$ $
d_T=\|M_T\|_F.$ Since both matrices have orthonormal columns, $d_T$ measures the extent to which the estimated cointegrating space is not orthogonal to the true common-trend space. Smaller values of $d_T$ therefore indicate that the misspecified estimated space is closer to the true cointegration space. In the Monte Carlo experiment, the event $d_{1000}<d_{100}$ occurs in 89.8\% of replications, supporting the interpretation above.

In general, four stylized facts emerge. First, the correct rank pair is the outcome that is the most frequently selected even for small sample sizes, and its frequency increases with the sample size. Second, the incorrect admissible pair is only rarely selected. Third, \emph{Und.\,1} tends to become more frequent as $T$ grows. Fourth, \emph{Und.\,2} is mainly a small-sample phenomenon and decreases substantially with larger samples. Hence, the proposed procedure performs satisfactorily in general, although undefined cases remain its main limitation. In such cases, one may rely on an auxiliary criterion, such as an information criterion or prior structural knowledge, to select among the remaining admissible pairs.

\subsection{Finite-Sample Performance of Hypothesis Testing}

Next, we study the finite-sample behavior of the proposed hypothesis tests for the ECC-MAR model. Unlike the previous Monte Carlo experiments, where the DGP parameters were randomly generated, here we consider a fixed design in order to assess size and power more directly. 

We consider an ECC-MAR model with $(m,n)=(4,3)$ and $(r_1,r_2)=(2,2)$, with parameters
\small
\[
\tau =
\begin{bmatrix}
0 & 0 \\
0 & 0 \\
-0.5 & 0 \\
0 & -0.5
\end{bmatrix},
\qquad
\gamma =
\begin{bmatrix}
-1 & 0 \\
1 & -0.5 \\
0.5 & -0.5 \\
0 & 1
\end{bmatrix},
\qquad
\phi =
\begin{bmatrix}
0 & 0 \\
-0.2 & 0 \\
0 & -0.2
\end{bmatrix},
\qquad
\theta =
\begin{bmatrix}
-1 & -0.5 \\
1 & 0.5 \\
0 & 1
\end{bmatrix}.
\]
\normalsize

For the size analysis, we consider three true null hypotheses: (i) the first row of $\tau$ is zero, corresponding to a weak exogeneity restriction; (ii) $[\,1\;\;1\;\;0\,]\theta = 0$; and (iii) $[\,0\;\;-0.5\;\;-0.5\;\;1\,]' \in \mathrm{span}(\gamma)$.
\begin{table}[!ht]
\centering
\renewcommand{\arraystretch}{1.2}
\begin{tabular}{lccc}
\hline
Null hypothesis & $T=100$ & $T=500$ & $T=1000$ \\
\hline
$\tau_{(1)}=0$ (weak exogeneity) & 5.6\% & 6\% & 5.8\% \\
$[\,1\;\;1\;\;0\,]\theta = 0$ & 8.8\% & 7.2\% & 4.8\% \\
$[\,0\;\;-0.5\;\;-0.5\;\;1\,]' \in \mathrm{span}(\gamma)$ & 7\% & 5\% & 4.8\% \\
\hline
\end{tabular}
\caption{Empirical rejection frequencies of the proposed hypothesis tests over 100 Monte Carlo replications. All tests are conducted at the 5\% nominal significance level.}
\label{tab:hypothesis_tests}
\end{table}
Table~\ref{tab:hypothesis_tests} shows that the tests have a satisfactory finite-sample size. For restrictions on cointegration matrices, rejection frequencies move toward the nominal level as $T$ increases. The weak exogeneity test is already well sized even in the smallest sample with an empirical rejection frequency that fluctuates around the nominal level and do not exhibit a systematic trend with \(T\).

We then examine power under false null hypotheses. Specifically, we test $\tau_{(3)}=0$, the linear restrictions $[\,1\;\;1.1\;\;0\,]\theta=0$, $[\,1\;\;1.3\;\;0\,]\theta=0$, and $[\,1\;\;1.5\;\;0\,]\theta=0$, and the membership restrictions $[\,0\;\;-0.5\;\;-0.6\;\;1\,]'$, $[\,0\;\;-0.5\;\;-0.8\;\;1\,]'$, and $[\,0\;\;-0.5\;\;-1\;\;1\,]' \in \mathrm{span}(\gamma)$.
\begin{table}[!ht]
\centering
\renewcommand{\arraystretch}{1.2}
\begin{tabular}{lccc}
\hline
Alternative hypothesis & $T=100$ & $T=500$ & $T=1000$ \\
\hline
$\tau_{(3)}=0$  & 100\% & 100\% &  100\%\\
\hline
$[\,1\;\;1.1\;\;0\,]\theta = 0$  & 69.8\% & 100\% & 100\% \\
$[\,1\;\;1.3\;\;0\,]\theta = 0$ & 97.8\%  & 100\% & 100\% \\
$[\,1\;\;1.5\;\;0\,]\theta = 0$  & 99.6\% & 100\% & 100\% \\
\hline
$[\,0\;\;-0.5\;\;-0.6\;\;1\,]' \in \mathrm{span}(\gamma)$  & 97.6\% & 100\% & 100\% \\
$[\,0\;\;-0.5\;\;-0.8\;\;1\,]' \in \mathrm{span}(\gamma)$  &  100\%& 100\% & 100\% \\
$[\,0\;\;-0.5\;\;-1\;\;1\,]' \in \mathrm{span}(\gamma)$ & 100\% & 100\% & 100\% \\
\hline
\end{tabular}
\caption{Empirical rejection frequencies (in percent) of the proposed hypothesis tests under false null hypotheses (power analysis) over 100 Monte Carlo replications. All tests are conducted at the 5\% nominal significance level.}
\label{tab:hypothesis_tests_power}
\end{table}
The power results in Table~\ref{tab:hypothesis_tests_power} are very strong. For $T\geq 500$, empirical power is one for all alternatives. At $T=100$, power is still 1 for the weak exogeneity test and above 97\% for the membership tests on $\mathrm{span}(\gamma)$. The only case in which power is more moderate is the restriction on $\theta$ under the closest alternative, $[\,1\;\;1.1\;\;0\,]\theta=0$, where the rejection frequency is 69.8\%. As expected, power increases monotonically as the alternative moves farther from the null, reaching 99.6\% for hypothesis $[\,1\;\;1.5\;\;0\,]\theta = 0$.

Overall, the simulations indicate that the proposed tests combine satisfactory size with very strong power, especially in moderate and large samples.

\section{Empirical Illustration}

To illustrate the empirical usefulness of the proposed C-MAR model, we consider a macroeconomic application with $m=2$ variables observed for $n=3$ European countries. Specifically, we study the joint dynamics of consumer inflation expectations ($\pi_e$) and industrial production ($\mathrm{In}$) for Belgium (B), Germany (D), and Austria (A).

From an economic perspective, stronger industrial production is expected to signal robust demand conditions, inducing firms to raise prices and consumers to revise inflation expectations upward. We therefore expect a long-run relation of the form $\pi_e=\varphi\,\mathrm{In}$ with $\varphi>0$. Across countries, the common monetary policy and the harmonized institutional environment suggest similar long-run dynamics in both inflation expectations and industrial production.

The analysis is based on monthly data from January 2000 to December 2019, thus excluding the Covid-19 period and the Russia--Ukraine war. Inflation expectations are proxied by consumer survey data from FRED, while industrial production indices are taken from Eurostat; all series are seasonally adjusted. Figures~\ref{F1} and~\ref{F2} plot the data.

\begin{figure}
    \centering
    \includegraphics[width=1\linewidth]{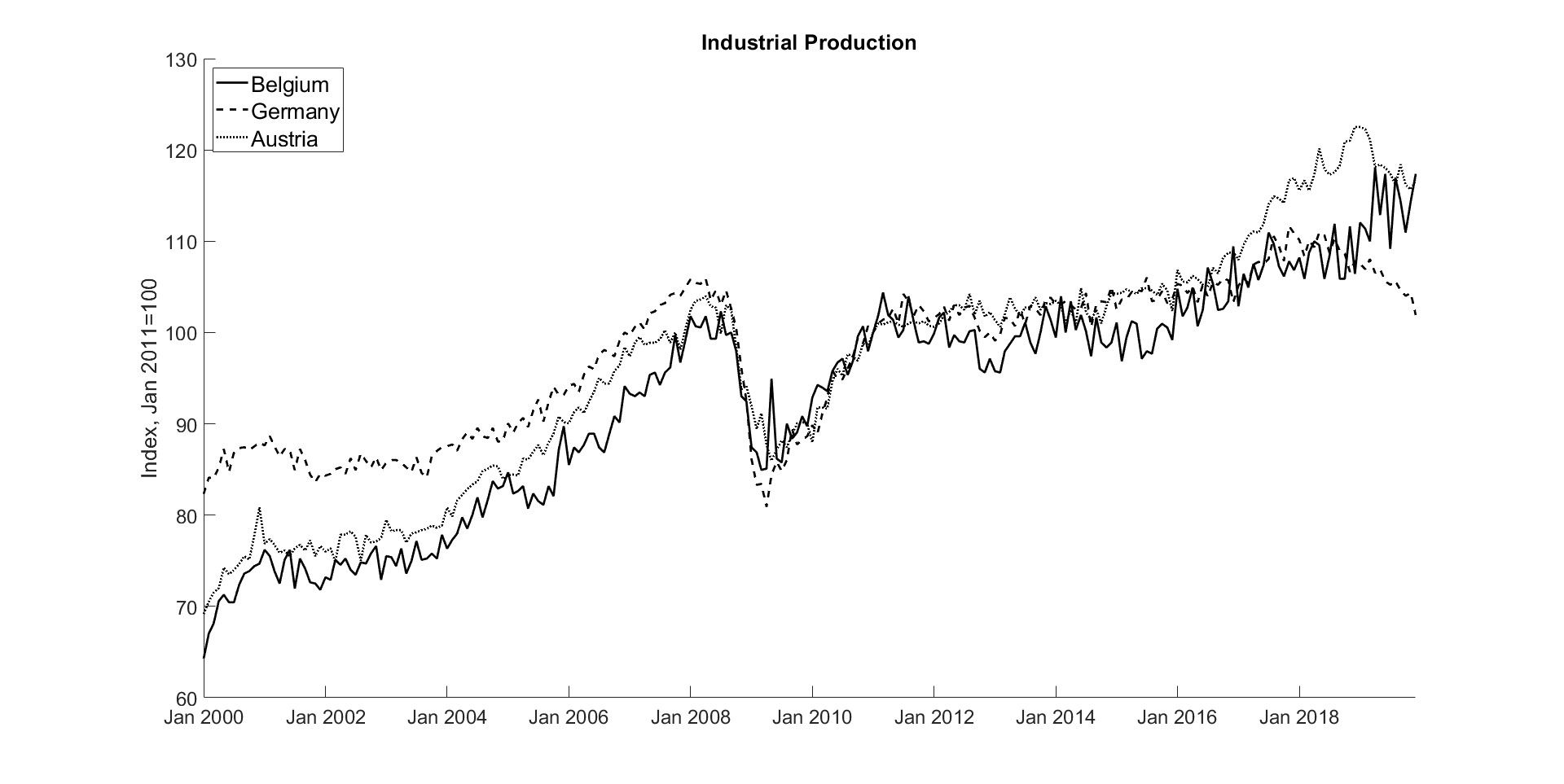}
    \caption{Index: Production in industry for Belgium, Germany, and Austria. Monthly data from Jan-2000 to Dec-2019. Seasonally Adjusted.
}
    \label{F1}
\end{figure}

\begin{figure}
    \centering
    \includegraphics[width=1\linewidth]{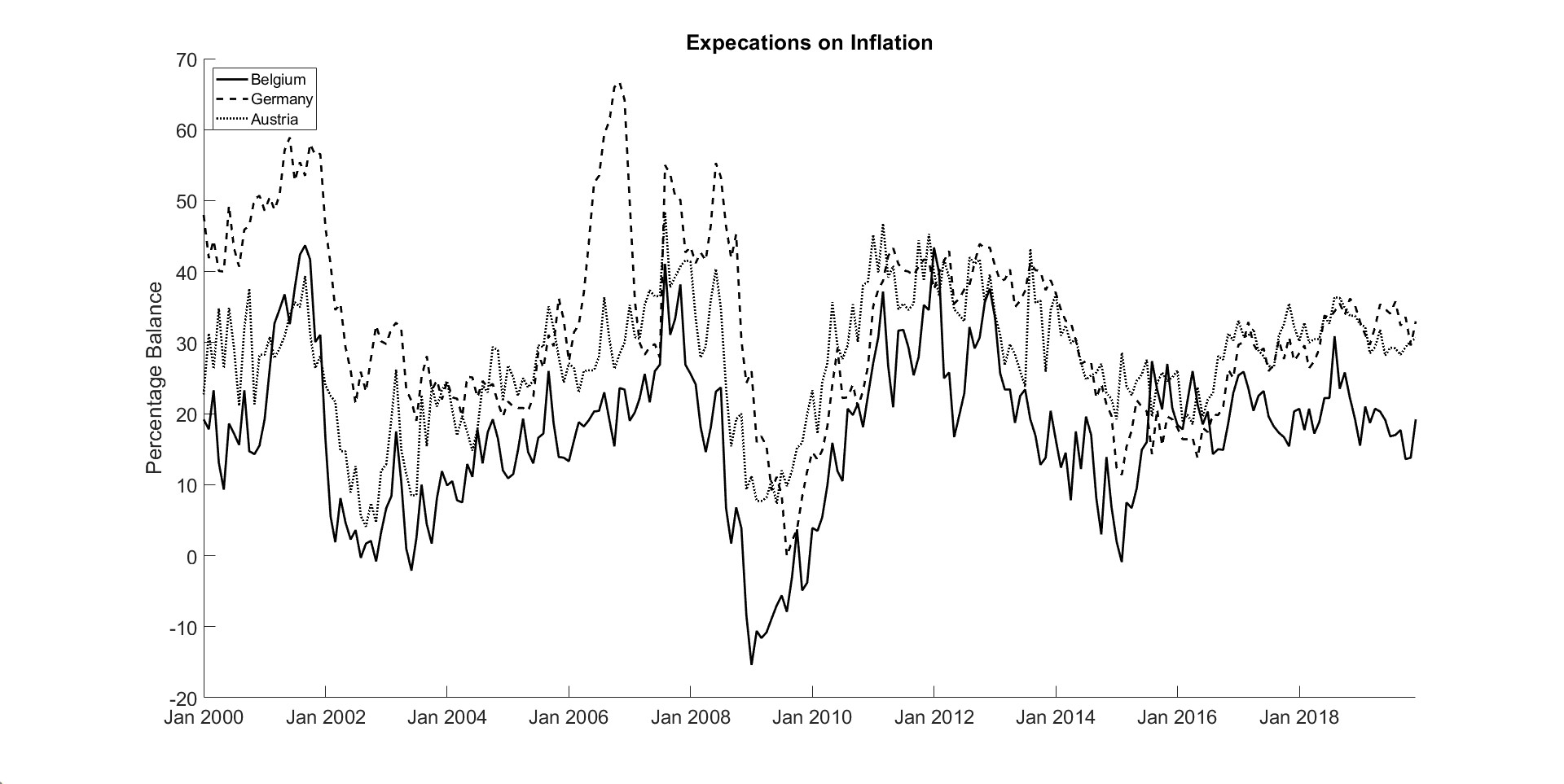}
    \caption{Consumer Opinion Surveys: Consumer Prices: Future Tendency for Belgium, Germany, and Austria. Monthly data from Jan-2000 to Dec-2019. Seasonally Adjusted.
}
    \label{F2}
\end{figure}

A preliminary analysis indicates that all six variables are $I(1)$, that the appropriate lag order is $p=1$, and that the total cointegration rank of the vectorized system is $r=4$. For brevity, the Dickey--Fuller tests, the BIC-based lag-order selection, and the Johansen trace test results are reported in Appendix~D.

By Theorem~\ref{T01}, the vectorized cointegration rank satisfies $r = nr_1 + mr_2 - r_1r_2.$ Since here $m=2$ and $r=4$, it follows that $r_1=1$ and $r_2=1$. Estimating the model by the procedure of Section~4 yields the unrestricted estimates
\begin{equation*}
    \hat{\gamma}=\begin{bmatrix}
        -0.294 \\
    1
    \end{bmatrix} 
     \qquad \hat{\tau} = \begin{bmatrix}
        0.0006 \\
   -0.0907
    \end{bmatrix}
    \qquad
    \hat{\theta}= \begin{bmatrix}
       -1 \\
    0.0535 \\
   0.912
    \end{bmatrix} \qquad \hat{\varphi}=\begin{bmatrix}
        0.3526 \\
    -0.0242 \\
    -0.0737
    \end{bmatrix}. 
\end{equation*}

These estimates imply a long-run relation of the form $\pi_e = 0.294\,\mathrm{In},$ which is consistent with the prior expectation of a positive long-run association between inflation expectations and industrial production.

We next test weak exogeneity by imposing row restrictions on the adjustment matrices $\tau$ and $\varphi$. The results are reported in Table~\ref{T9}.

Consistent with our hypotheses, the results reveal a long-term relationship between price expectations and industrial production of the form $\pi_e = 0.294 \text{In}$. The positive coefficient confirms the anticipated association.  

Next, we test for restrictions on the adjustment matrices $\varphi$ and $\tau$ to determine which variable adjusts in response to a disequilibrium. In other words, we test whether the $i^{\text{th}}$ row of $\varphi$ and $\tau$ is null. The results of the weak exogeneity test are recorded in table \ref{T9}.

\begin{table}[h!]
\centering
\begin{tabular}{@{}ccc@{}}
\toprule
\textbf{Row} & \textbf{Statistic} & \textbf{p-value} \\ \midrule
$\pi_{e}$ & 20.1937 & 6.9984e-06 \\ 
$\text{In}$ & 0.0086 & 0.9260\\ 
$B$ & 35.5164 & 2.5291e-09 \\ 
$D$ & 0.6053 & 0.4366 \\ 
$A$ & 4.8998 & 0.0269
 \\ \bottomrule
\end{tabular}
\caption{Weak Exogeneity Test Results. The table presents the test statistics and corresponding p-values for testing the null hypothesis of weak exogeneity for various parameters (rows) in the model.}
\label{T9}
\end{table}

The results indicate that industrial production is weakly exogenous with respect to $\gamma$, while inflation expectations adjust to restore the equilibrium relation. This is consistent with the interpretation of industrial production as the driving stochastic trend and inflation expectations as the adjusting variable.

Turning to the cross-country dimension, Germany appears to be weakly exogenous, suggesting that it acts as the driving trend within the country equilibrium relation, while Belgium and Austria adjust to restore the long-run balance. This is economically plausible given Germany’s dominant role in the euro-area economy.

We next assess the adequacy of the estimated C-MAR specification by examining whether the estimated long-run relations are indeed stationary. If $\hat\gamma$ and $\hat\theta$ correctly identify the row-wise and column-wise cointegration spaces, then the transformed processes $\hat\gamma'X_t$ and $X_t\hat\theta$ should be $I(0)$, even though the entries of $X_t$ are individually $I(1)$. We test stationarity of the five resulting linear combinations with a Dickey--Fuller test. The results, reported in Table~\ref{T7}, reject the unit-root null in all five cases, supporting the interpretation of $\hat\gamma'X_t=0$ and $X_t\hat\theta=0$ as equilibrium relations.
\begin{table}[h!]
\centering
\begin{tabular}{@{}ccc@{}}
\toprule
\textbf{Row} & \textbf{Statistic} & \textbf{p-value} \\ \midrule
$\pi_{e,B} -0.294 \text{In}_\text{B}$ & -2.6723 & 0.008 \\ 
$\pi_{e,D} -0.294 \text{In}_\text{D}$ & \text{-2.8406
} & \text{0.0048
} \\ 
$\pi_{e,A} -0.294 \text{In}_\text{A}$ & \text{-4.8443} & \text{$<$1.0000e-03
} \\ 
$\text{In}_\text{B} -0.0535 \text{In}_\text{D} -0.912 \text{In}_\text{A}$ & \text{-7.8869} & \text{$<$1.0000e-03} \\ 
$\pi_{e,B} -0.0535 \pi_{e,D} -0.912 \pi_{e,A}$ & \text{-3.9111} & \text{$<$1.0000e-03} \\ \bottomrule
\end{tabular}
\caption{Dickey-Fuller Test Results for the estimated cointegration relationships: $\hat{\gamma}'X_t$ and $X_t \hat{\theta}$. The table reports the test statistic and corresponding p-value computed using the test specification without a constant or trend term.}
\label{T7}
\end{table}

Overall, the C-MAR specification provides a coherent description of the data by jointly modeling row and column cointegration. The empirical results support a positive long-run relationship between inflation expectations and industrial production, with industrial production acting as the common driving trend. At the country level, the estimates point to a long-run equilibrium across countries, with Germany behaving as the dominant stochastic trend and Belgium and Austria adjusting toward equilibrium.

\subsection{Comparison with Standard VECM}

For comparison, we estimate a standard VECM for the six vectorized series. The estimated model is
\small
\begin{equation*}
    \Delta \begin{bmatrix}
        \text{In}_B \\
        \text{In}_D \\
        \text{In}_A \\
        \pi_{e,B} \\
        \pi_{e,D} \\
        \pi_{e,A}
    \end{bmatrix}_t=
    \begin{bmatrix}
    3.1   & -4.1  & -8.9  & -2.2 \\
   -0.1  &  3 &   1.3 &    2.3 \\
   -0.8  &  0.2  & -2.1  & 0.5 \\
    0.2 & -12.6 &  -9.1 &  -2.4 \\
   -0.5  &  0.7 & -11.7  &  0.1 \\
   -0.5  &  1 &   -1.8 &   -8.5  \\
    \end{bmatrix}
    \begin{bmatrix}
    -0.95  & -0.2  &  1 &  0 & 0& 0 \\
    2.08  & -3.59  & 0& 1& 0&0 \\
    3.28 & -4.88&0& 0& 1& 0 \\
    1.37 &-2.59 & 0& 0& 0& 1
    \end{bmatrix}\begin{bmatrix}
        \text{In}_B \\
        \text{In}_D \\
        \text{In}_A \\
        \pi_{e,B} \\
        \pi_{e,D} \\
        \pi_{e,A}
    \end{bmatrix}_{t-1}+\varepsilon_t
\end{equation*}
\normalsize

Unlike the C-MAR model, the unrestricted VECM does not separate the row and column cointegration mechanisms, making economic interpretation substantially more difficult. In particular, the signs and magnitudes linking industrial production and inflation expectations vary markedly across cointegrating vectors and across normalizations, so that the implied long-run relationships are not stable across countries. For example, one normalization may suggest a positive association between industrial production and inflation expectations in one country, while another normalization implies a negative relation for another country. This lack of robustness makes the economic interpretation of the unrestricted VECM considerably less transparent than that of the matrix model.

This instability is not surprising. As shown in Section~3, the unrestricted vector representation mixes the two underlying structures, whereas the C-MAR model isolates them explicitly through separate row-wise and column-wise cointegration components. As a result, the VECM coefficients do not admit a decomposition with the same structural interpretation as in the matrix specification.

The same issue arises when considering the adjustment mechanism. While the C-MAR model identifies industrial production as the driving force behind inflation expectations and Germany as the driving force in the cross-country relation, the unrestricted VECM does not yield a similarly clear pattern. In fact, the null of a zero row in the adjustment matrix is rejected for all six variables, as shown in Table~\ref{T11}.

\begin{table}[h!]
\centering
\begin{tabular}{@{}ccc@{}}
\toprule
\textbf{Row} & \textbf{Statistic} & \textbf{p-value} \\ \midrule
In$_B$ & 50.1537 & 3.3536e-10 \\ 
In$_D$ & 11.8111 & 0.0188\\ 
In$_A$ & 17.9833 & 0.0012 \\ 
$\pi_{e,B}$ &18.1288 & 0.0012 \\ 
$\pi_{e,D}$ & 17.9159& 0.0013 \\ 
$\pi_{e,A}$ & 33.2618 & 1.0558e-06
 \\ \bottomrule
\end{tabular}
\caption{Weak Exogeneity Test Results. The table presents the test statistics and corresponding p-values for testing the null hypothesis of weak exogeneity.}
\label{T11}
\end{table}

A further distinction concerns the interpretation of weak exogeneity itself. In the C-MAR framework, weak exogeneity is formulated at the structural level, that is, for entire row or column components such as industrial production, inflation expectations, or country effects. By contrast, in the standard VECM it can only be tested series by series. Hence, one cannot test whether inflation expectations as a whole are weakly exogenous, but only whether each of $\pi_{e,B}$, $\pi_{e,D}$, and $\pi_{e,A}$ is weakly exogenous separately.

Overall, the comparison highlights the main empirical advantage of the C-MAR approach: by preserving the matrix structure of the data, it yields a much more coherent and economically interpretable description of both the long-run relations and the adjustment dynamics than the unrestricted VECM.

\section{Conclusions}

This paper extends the analysis of matrix autoregressive models to cointegrated systems, thereby contributing to the growing literature on models for matrix-valued observations.

We argue that the use of a matrix-based model, rather than a standard vector autoregressive framework, is motivated by two main advantages: a reduction in estimation complexity and a clearer interpretation of the dynamics through the separation of row-wise and column-wise effects. Although previous contributions on cointegration in MAR settings addressed the first objective, they did not provide cointegration coefficients with a clear and economically meaningful interpretation within the MAR framework.

To fill this gap, we propose a novel cointegrated MAR model featuring parameters with direct economic content: cointegration matrices, which describe the long-run equilibrium relations toward which the system is attracted, and adjustment matrices, which measure the speed at which deviations from equilibrium are corrected. The presence of interpretable coefficients makes it possible to formulate economically meaningful restrictions and to test them using standard statistical procedures.

In contrast to earlier work, we do not start from an error-correction specification whose coefficients admit a Kronecker-product decomposition. Instead, we begin with a MAR representation and then derive the corresponding error-correction form. A key strength of our approach is that the matrix structure is preserved in both the autoregressive and the error-correction representations. By contrast, when the error-correction representations proposed in earlier contributions are vectorized, the resulting autoregressive coefficients generally do not admit a Kronecker-product decomposition.

\bibliographystyle{apalike}
\bibliography{sample}

\appendix

\section{Proofs}

\begin{proof}[\textbf{Proof of Theorem \ref{T01}}]
The MAR(1) model can be vectorized as
\[
x_t=\text{vec}(X_t)= (\Psi\otimes\Lambda)\text{vec}(X_{t-1})+\text{vec}(E_t)
=Kx_{t-1}+e_t,
\]
where $K=\Psi\otimes\Lambda$.

Set
\[
c_1:=m-r_1,\qquad c_2:=n-r_2,\qquad c:=c_1c_2.
\]
Since the eigenvalue $1$ of $\Lambda$ is semisimple with algebraic multiplicity $c_1$,
there exists an invertible matrix $S$ such that
\[
S^{-1}\Lambda S=
\begin{pmatrix}
I_{c_1} & 0\\
0 & A
\end{pmatrix},
\]
where every eigenvalue of $A$ has modulus strictly smaller than one.
Similarly, since the eigenvalue $1$ of $\Psi$ is semisimple with algebraic multiplicity $c_2$,
there exists an invertible matrix $T$ such that
\[
T^{-1}\Psi T=
\begin{pmatrix}
I_{c_2} & 0\\
0 & B
\end{pmatrix},
\]
where every eigenvalue of $B$ has modulus strictly smaller than one.

Therefore,
\[
K=\Psi\otimes\Lambda
=(T\otimes S)
\left[
\begin{pmatrix}
I_{c_2} & 0\\
0 & B
\end{pmatrix}
\otimes
\begin{pmatrix}
I_{c_1} & 0\\
0 & A
\end{pmatrix}
\right]
(T\otimes S)^{-1}.
\]
Hence $K$ is similar to
\[
\widetilde K:=P\otimes L=
\begin{pmatrix}
I_c & 0 & 0 & 0\\
0 & I_{c_2}\otimes A & 0 & 0\\
0 & 0 & B\otimes I_{c_1} & 0\\
0 & 0 & 0 & B\otimes A
\end{pmatrix},
\]
where $\widetilde K$ is obtained from $P\otimes L$ by a simultaneous permutation of rows and columns. It follows that the eigenvalues of $K$ are:
\begin{itemize}
    \item exactly $c=(m-r_1)(n-r_2)$ eigenvalues equal to $1$, coming from the block $I_c$;
    \item all remaining eigenvalues belonging to the blocks
    $B\otimes I_{c_1}$, $I_{c_2}\otimes A$, and $B\otimes A$.
\end{itemize}
Since every eigenvalue of $A$ and $B$ has modulus strictly smaller than one, every eigenvalue
of these three blocks also has modulus strictly smaller than one.

Now consider the characteristic polynomial of the vectorized VAR(1),
\[
\Phi(z)=I_{mn}-zK.
\]
Because the roots of $\det\Phi(z)=0$ are of the form $z=\nu^{-1}$, where $\nu$ is an eigenvalue of $K$,
we obtain:
\begin{itemize}
    \item if $\nu=1$, then $z=1$;
    \item if $|\nu|<1$, then $|z|=|\nu|^{-1}>1$.
\end{itemize}
Therefore
\[
\det\Phi(z)=0 \quad \Longrightarrow \quad z=1 \ \text{or} \ |z|>1.
\]
Moreover, since $K$ is similar to a block diagonal matrix with the block $I_c$,
the eigenvalue $1$ of $K$ is semisimple.

Hence the vectorized system satisfies the Johansen $I(1)$ condition for a simple VAR(1) (see Theorem 4.2 in \cite{johansen1995likelihood}):
the only unit root is at $z=1$, it is of first order, and all other roots lie outside the unit circle.
Equivalently,
\[
\Delta x_t = \Pi x_{t-1}+e_t,
\qquad
\Pi:=K-I_{mn},
\]
is a cointegrated error-correction representation.

Finally, since
\[
\dim\ker(\Pi)=\dim\ker(K-I_{mn})=c=(m-r_1)(n-r_2),
\]
the rank-nullity theorem yields
\[
\text{rank}(\Pi)=mn-c=mn-(m-r_1)(n-r_2).
\]
Therefore the vectorized process $x_t$ is $I(1)$ and cointegrated in the sense of Johansen (see Theorem 4.2 in \cite{johansen1995likelihood}),
with cointegration rank
\[
mn-(m-r_1)(n-r_2),
\]
and therefore
\[
(m-r_1)(n-r_2)
\]
common stochastic trends.
\end{proof}

\begin{proof}[\textbf{Proof of Lemma \ref{L1}}]

We first show that a necessary condition for
\[
\Psi \otimes \Lambda - I = A^* \otimes B^*
\]
is that either \(\Psi\) or \(\Lambda\) is a scalar multiple of the identity.

Suppose that there exist matrices
\[
A^* \in \mathbb{R}^{n\times n},
\qquad
B^* \in \mathbb{R}^{m\times m}
\]
such that
\[
\Psi \otimes \Lambda - I_{mn} = A^* \otimes B^*.
\]
Since
\[
I_{mn} = I_n \otimes I_m,
\]
this can be rewritten as
\[
\Psi \otimes \Lambda - I_n \otimes I_m = A^* \otimes B^*.
\]

Assume that \(\Psi\) is not a scalar multiple of \(I_n\). Then \(\Psi\) and \(I_n\) are linearly independent in \(\mathbb{R}^{n\times n}\). Therefore, there exist linear functionals
\[
f,g : \mathbb{R}^{n\times n} \to \mathbb{R}
\]
such that
\[
f(\Psi)=1,\qquad f(I_n)=0,
\]
and
\[
g(\Psi)=0,\qquad g(I_n)=1.
\]

Now apply the linear map
\[
f\otimes \mathrm{id}_{M_m}
\]
to both sides of
\[
\Psi \otimes \Lambda - I_n \otimes I_m = A^* \otimes B^*,
\]
where \(\mathrm{id}_{M_m}\) denotes the identity map on \(\mathbb{R}^{m\times m}\). Using the identity
\[
(f\otimes \mathrm{id}_{M_m})(X\otimes Y)=f(X)\,Y,
\]
we obtain
\[
\Lambda = f(A^*)\,B^*.
\]

Similarly, applying \(g\otimes \mathrm{id}_{M_m}\) yields
\[
-I_m = g(A^*)\,B^*.
\]

It follows that \(B^*\) is a scalar multiple of \(I_m\), and hence the previous identity implies that \(\Lambda\) is also a scalar multiple of \(I_m\). Therefore, if \(\Psi\) is not a scalar multiple of the identity, then \(\Lambda\) must be a scalar multiple of the identity.

We now show that Assumption \ref{Ass_1} implies that neither $\Psi$ nor $\Lambda$
is a scalar multiple of the identity. Indeed, suppose by contradiction that $\Psi=aI_n$ for some $a\in\mathbb{R}$.
Then all eigenvalues of $\Psi$ are equal to $a$. Hence the eigenvalue $1$ has either
multiplicity $n$ (if $a=1$) or multiplicity $0$ (if $a\neq 1$). This contradicts
Assumption \ref{Ass_1}, which requires the eigenvalue $1$ of $\Psi$ to have multiplicity
$n-r_2$, with $0<r_2<n$. 

The same argument shows that
$\Lambda$ cannot be a scalar multiple of the identity.

\end{proof}

\begin{proof}[\textbf{Proof of Theorem \ref{Teo_3}}]
We first consider $Y_t=\gamma'X_t$. From
\[
X_t=\Lambda X_{t-1}\Psi'+E_t
\]
and the decomposition $\Lambda=I_m+\tau\gamma'$, we obtain
\begin{align*}
Y_t
&=\gamma'X_t \\
&=\gamma'\Lambda X_{t-1}\Psi' + \gamma'E_t \\
&=\bigl(\gamma' + \gamma'\tau\gamma'\bigr)X_{t-1}\Psi' + \gamma'E_t \\
&=\bigl(I_{r_1}+\gamma'\tau\bigr)\gamma'X_{t-1}\Psi' + \gamma'E_t \\
&=: A_\gamma Y_{t-1}\Psi' + U_t,
\end{align*}
where
\[
A_\gamma := I_{r_1}+\gamma'\tau,
\qquad
U_t:=\gamma'E_t.
\]
Hence $\{Y_t\}$ is itself a MAR(1) process.

Since $\tau\gamma'$ and $\gamma'\tau$ have the same nonzero eigenvalues, the eigenvalues of
$A_\gamma=I_{r_1}+\gamma'\tau$ are exactly the non-unit eigenvalues of
$\Lambda=I_m+\tau\gamma'$. By assumption, all such eigenvalues have modulus strictly smaller
than one, and therefore
\[
\rho(A_\gamma)<1.
\]
where $\rho(\cdot)$ denotes the spectral radius. Moreover, $\Psi$ has unit eigenvalues and all remaining eigenvalues strictly inside the unit disk,
hence
\[
\rho(\Psi)=1.
\]
Therefore
\[
\rho(A_\gamma)\rho(\Psi)<1.
\]
By Proposition 1 in Chen, Xiao and Yang (2021), the MAR(1) process
\[
Y_t=A_\gamma Y_{t-1}\Psi'+U_t
\]
is stationary and causal. Hence $Y_t=\gamma'X_t$ is $I(0)$. 

We now consider $Z_t=X_t\theta$. Using $\Psi=I_n+\varphi\theta'$, we have
\begin{align*}
Z_t
&=X_t\theta \\
&=\Lambda X_{t-1}\Psi'\theta + E_t\theta \\
&=\Lambda X_{t-1}\bigl(\theta+\theta\varphi'\theta\bigr) + E_t\theta \\
&=\Lambda X_{t-1}\theta\bigl(I_{r_2}+\varphi'\theta\bigr) + E_t\theta \\
&=: \Lambda Z_{t-1} B_\theta' + V_t,
\end{align*}
where
\[
B_\theta' := I_{r_2}+\varphi'\theta,
\qquad
V_t:=E_t\theta.
\]
Thus $\{Z_t\}$ is again a MAR(1) process.

Now $\Psi' = I_n+\theta\varphi'$, and $\theta\varphi'$ and $\varphi'\theta$ have the same nonzero
eigenvalues. Hence the eigenvalues of
\[
B_\theta' = I_{r_2}+\varphi'\theta
\]
are exactly the non-unit eigenvalues of $\Psi'$, and therefore of $\Psi$. By assumption, they all
have modulus strictly smaller than one, so
\[
\rho(B_\theta')<1.
\]
Since $\Lambda$ has unit eigenvalues and all other eigenvalues strictly inside the unit disk,
\[
\rho(\Lambda)=1.
\]
Therefore
\[
\rho(\Lambda)\rho(B_\theta')<1.
\]
Applying again Proposition 1 in Chen, Xiao and Yang (2021), the MAR(1) process
\[
Z_t=\Lambda Z_{t-1}B_\theta' + V_t
\]
is stationary and causal. Hence $Z_t=X_t\theta$ is $I(0)$. 

This proves the claim.
\end{proof}

\begin{proof}[\textbf{Proof of Theorem \ref{T2}}]
Starting from \eqref{VECMECC-MAR}, we observe that by choosing
\[
H = (\theta_\perp \otimes \gamma_\perp),
\]
it holds that $\Pi H = 0$. Since $\Pi=\alpha\beta'$ with $\alpha$ of full column rank,
this implies $\mathcal N(\Pi)=\mathcal N(\beta')$ and hence $\beta' H=0$, where $\mathcal N(\cdot)$ denotes the null space of a matrix.
We observe that $H$ is of full column rank $(n-r_2)(m-r_1)$ as a consequence
of the definition of the orthogonal complement.
Following Theorem \ref{T01}, it follows that the rank of $\beta_\perp$ is exactly
$(n-r_2)(m-r_1)$, so that $\beta = H_\perp$.
Since
\[
\left[ (I_n \otimes \gamma), (\theta \otimes \gamma_\perp) \right]' H = 0,
\]
one can legitimately select
\[
\beta = \left[ (I_n \otimes \gamma), (\theta \otimes \gamma_\perp) \right].
\]
 
Given the choice of $\beta$, the corresponding adjustment matrix
can be obtained as $\alpha = \Pi \bar\beta = \Pi \beta(\beta'\beta)^{-1}$.
Partition $\beta$ as $\beta=[\beta_1,\beta_2]$ with $\beta_1=I_n\otimes\gamma$, $\beta_2=\theta\otimes\gamma_\perp$.
Since $\gamma'\gamma_\perp=0$, it follows that $\beta_1'\beta_2=0$, so that
$\beta'\beta$ is block diagonal:
\[
\beta'\beta=\begin{bmatrix}
    I_n\otimes(\gamma'\gamma) & 0 \\ 0 & (\theta'\theta)\otimes(\gamma_\perp'\gamma_\perp)
\end{bmatrix},
\]
hence 
\[
\alpha=
\Big[
\Pi \beta_1 \big(I_n\otimes(\gamma'\gamma)\big)^{-1},
\;
\Pi \beta_2 \big((\theta'\theta)\otimes(\gamma_\perp'\gamma_\perp)\big)^{-1}
\Big].
\]
Using the expression $\Pi = I_n\otimes(\tau\gamma') + (\varphi\theta')\otimes I_m + (\varphi\theta')\otimes(\tau\gamma')$, we obtain
\[
\Pi\beta_1 (I_n\otimes(\gamma'\gamma)^{-1})
= I_n\otimes\tau
+ (\varphi\theta')\otimes(\bar\gamma+\tau),
\]
where $\bar\gamma=\gamma(\gamma'\gamma)^{-1}$.
Similarly,
\[
\Pi\beta_2\big((\theta'\theta)^{-1}\otimes(\gamma_\perp'\gamma_\perp)^{-1}\big)
= \varphi\otimes\bar\gamma_\perp.
\]
This proves the theorem.
\end{proof}

\begin{proof}[\textbf{Proof of Lemma \ref{L2}}]
    Assume $A_i=\lambda_i A_1$ for all $i$. Then, by bilinearity of the Kronecker product, $\sum_{i=1}^n A_i\otimes B_i
=
\sum_{i=1}^n (\lambda_i A_1)\otimes B_i
=
A_1\otimes\left(\sum_{i=1}^n \lambda_i B_i\right)$, which proves the first statement. The second statement follows analogously by
assuming $B_i=\delta_i B_1$ and factoring out $B_1$.
\end{proof}

\begin{proof}[\textbf{Proof of Lemma \ref{L3}}]
If $A_i=\lambda_i A_1$ and $B_i=\delta_i B_1$ for some scalars
$\lambda_i,\delta_i$, then
\[
\sum_{i=1}^n A_i\otimes B_i
=
\sum_{i=1}^n (\lambda_i A_1)\otimes(\delta_i B_1)
=
A_1\otimes\Big(B_1\sum_{i=1}^n \lambda_i\delta_i\Big)
=
A_1\otimes(B_1 s),
\]
so that $C=A_1$ and $D=B_1 s$. Since $\text{eig}(B_1 s)=s\,\text{eig}(B_1)$, we have
$|\text{eig}(D)|=|s|\,|\text{eig}(B_1)|\le 1$. 
\end{proof}

\begin{proof}[\textbf{Proof of Lemma \ref{L4}}]
Set
\[
P:=\varphi\theta' \in \mathbb{R}^{n\times n},
\qquad
Q:=\tau\gamma' \in \mathbb{R}^{m\times m}.
\]

By a straightforward extension of the proof of Lemma \ref{L1}, a necessary condition for the existence of matrices \(A^*\) and \(B^*\) such that
\[
P\otimes Q + I_{mn} = A^* \otimes B^*
\]
is that either \(P\) or \(Q\) is a scalar multiple of the identity.

Assumption \ref{ASS_2} implies that
\[
\operatorname{rank}(P)=r_2,\qquad \operatorname{rank}(Q)=r_1,
\qquad
0<r_2<n,\qquad 0<r_1<m.
\]
If \(P=aI_n\) for some \(a\in\mathbb{R}\), then either \(a=0\), in which case
\[
\operatorname{rank}(P)=0,
\]
or \(a\neq 0\), in which case
\[
\operatorname{rank}(P)=n.
\]
Both cases contradict \(0<r_2<n\). Hence \(P\) is not a scalar multiple of \(I_n\). By the same argument, \(Q\) is not a scalar multiple of \(I_m\).

\end{proof}
\begin{proof}[\textbf{Proof of Theorem \ref{Teo_4}}]

We first show that $\gamma'X_t$ and $X_t\theta$ have stationary first differences.
Premultiplying \eqref{MECM} by $\gamma'$ gives
\[
\Delta(\gamma'X_t)
=
\gamma'\tau\,\gamma'X_{t-1}\theta\,\varphi' + \gamma'E_t.
\]
Since $\gamma'X_{t-1}\theta$ is $I(0)$ and $E_t$ is white noise, it follows that
\[
\Delta(\gamma'X_t)
\]
is $I(0)$. Therefore, $\gamma'X_t$ is at most $I(1)$.

Similarly, postmultiplying \eqref{MECM} by $\theta$ yields
\[
\Delta(X_t\theta)
=
\tau\,\gamma'X_{t-1}\theta\,\varphi'\theta + E_t\theta.
\]
Again, since $\gamma'X_{t-1}\theta$ is $I(0)$ and $E_t$ is white noise, we obtain
\[
\Delta(X_t\theta)
\]
is $I(0)$. Hence $X_t\theta$ is at most $I(1)$.

It remains to show that neither $\gamma'X_t$ nor $X_t\theta$ is $I(0)$.
Since $x_t$ is $I(1)$ with cointegration matrix $\beta=\theta\otimes\gamma$, by the
Granger representation theorem it admits a common-trend representation whose nonstationary
space is spanned by a complement of $\beta$. A convenient choice is
\[
\beta_\perp
=
\Big[
\theta_\perp\otimes I_m
\;,\;
\theta\otimes\gamma_\perp
\Big],
\]
Therefore, the columns of $\beta_\perp$ span the common-trend space of $x_t$.

Now define
\[
y_t:=\text{vec}(\gamma'X_t)=(I_n\otimes\gamma')x_t.
\]
Applying the transformation $I_n\otimes\gamma'$ to the common-trend space gives
\[
(I_n\otimes\gamma')\beta_\perp
=
\Big[
\theta_\perp\otimes\gamma'
\;,\;
\theta\otimes\gamma'\gamma_\perp
\Big]
=
\Big[
\theta_\perp\otimes\gamma'
\;,\;
0
\Big].
\]
Since $\gamma$ has full column rank and $\theta_\perp$ has rank $n-r_2>0$, the matrix
\[
\theta_\perp\otimes\gamma'
\]
has rank
\[
r_1(n-r_2)>0.
\]
Hence $y_t=\text{vec}(\gamma'X_t)$ retains nonstationary common-trend components and therefore
cannot be $I(0)$. Since we already proved that $\Delta(\gamma'X_t)$ is $I(0)$, it follows
that $\gamma'X_t$ is $I(1)$.

Next define
\[
z_t:=\text{vec}(X_t\theta)=(\theta'\otimes I_m)x_t.
\]
Applying the transformation $\theta'\otimes I_m$ to the common-trend space yields
\[
(\theta'\otimes I_m)\beta_\perp
=
\Big[
\theta'\theta_\perp\otimes I_m
\;,\;
\theta'\theta\otimes\gamma_\perp
\Big]
=
\Big[
0
\;,\;
\theta'\theta\otimes\gamma_\perp
\Big].
\]
Since $\theta$ has full column rank, $\theta'\theta$ is nonsingular, and $\gamma_\perp$
has rank $m-r_1>0$. Therefore,
\[
\theta'\theta\otimes\gamma_\perp
\]
has rank
\[
r_2(m-r_1)>0.
\]
Thus $z_t=\text{vec}(X_t\theta)$ also retains nonstationary common-trend components and cannot
be $I(0)$. Since $\Delta(X_t\theta)$ is $I(0)$, it follows that $X_t\theta$ is $I(1)$.

\end{proof}

\begin{proof}[\textbf{Proof of Proposition \ref{P_rankpairs}}]
If $\gamma^{\ast\prime}X_t$ is $I(0)$ componentwise, then every column of $\gamma^\ast$
belongs to $\mathcal C_\gamma$, implying $r_1^\ast\le \dim(\mathcal C_\gamma)=r_1$.
Similarly, if $X_t\theta^\ast$ is $I(0)$ componentwise, then every column of $\theta^\ast$
belongs to $\mathcal C_\theta$, implying $r_2^\ast\le \dim(\mathcal C_\theta)=r_2$.
Hence, if both $\gamma^{\ast\prime}X_t$ and $X_t\theta^\ast$ were $I(0)$ componentwise,
we would have $r_1^\ast\le r_1$ and $r_2^\ast\le r_2$, which together with the rank
identity forces $(r_1^\ast,r_2^\ast)=(r_1,r_2)$. Therefore, for any
$(r_1^\ast,r_2^\ast)\neq(r_1,r_2)$, at least one of the two conditions must fail.
\end{proof}

\section{Estimation Steps}

This appendix describes the estimation of the row-side parameters in
\eqref{Auxiliary_1} following the standard Johansen maximum likelihood logic, adapted
to the present matrix setting. The derivation is the same as in the vector I(1) model,
except that here all sample moments are pooled over both the time index $t$ and the
column index $j$ of the whitened auxiliary system. See Chapter 7 in \citet{juselius2006cointegrated} for the vector case, and \citet{li2024cointegrated} for the
same pooled-column idea in matrix autoregressive models. 

Consider the whitened auxiliary regression
\begin{equation}
\widetilde Y_t^{L}
=
\tau\gamma'\widetilde X_t^{L}
+
\sum_{i=1}^{p-1}\Gamma_{1,i}\widetilde Z_{i,t}^{L}
+
\widetilde E_t^{L},
\qquad
\widetilde E_t^{L}\sim MN(0,\Sigma_r,I),
\label{B.1}
\end{equation}
where $\widetilde Y_t^{L},\widetilde X_t^{L}\in\mathbb{R}^{m\times q_2}$ and
$q_2=n-r_2$. Reading \eqref{B.1} by columns gives, for $j=1,\dots,q_2$,
\begin{equation}
\widetilde y_{tj}^{L}
=
\tau\gamma'\widetilde x_{tj}^{L}
+
\Psi_1 \widetilde z_{tj}^{L}
+
\widetilde\varepsilon_{tj}^{L},
\label{B.2}
\end{equation}
where $\Psi_1$ stacks the short-run matrices $\Gamma_{1,1},\dots,\Gamma_{1,p-1}$.

\paragraph{Step 1. Concentrating out the short-run dynamics.}
Following the Frisch--Waugh theorem, regress $\widetilde y_{tj}^{L}$ and
$\widetilde x_{tj}^{L}$ on the short-run regressors $\widetilde z_{tj}^{L}$ over all
pooled observations $(t,j)$. Let the corresponding residuals be denoted by
$r_{0,tj}^{L}$ and $r_{1,tj}^{L}$. Then the concentrated long-run model is
\begin{equation}
r_{0,tj}^{L}
=
\tau\gamma' r_{1,tj}^{L}
+
u_{tj}^{L}.
\label{B.3}
\end{equation}

\paragraph{Step 2. Pooled sample moment matrices.}
Let $N_L=T_0 q_2$ denote the total number of pooled observations. Define
\begin{equation}
S_{00}^{L}
=
\frac{1}{N_L}\sum_{t}\sum_{j} r_{0,tj}^{L} r_{0,tj}^{L\prime},
\qquad
S_{11}^{L}
=
\frac{1}{N_L}\sum_{t}\sum_{j} r_{1,tj}^{L} r_{1,tj}^{L\prime},
\label{B.4}
\end{equation}
and
\begin{equation}
S_{01}^{L}
=
\frac{1}{N_L}\sum_{t}\sum_{j} r_{0,tj}^{L} r_{1,tj}^{L\prime},
\qquad
S_{10}^{L}=(S_{01}^{L})'.
\label{B.5}
\end{equation}
These are the pooled counterparts of the usual Johansen moment matrices. The only
difference from the vector case is the additional summation over the column index $j$.
This pooling is justified because, after whitening, \eqref{B.2} is a collection of
standard Gaussian reduced-rank regressions with common coefficient matrix
$\tau\gamma'$.

\paragraph{Step 3. Obtaining the cointegration matrix.}
The unrestricted maximum likelihood estimator of the cointegration matrix $\gamma$ is
obtained from the generalized eigenvalue problem
\begin{equation}
\left|
\lambda S_{11}^{L}
-
S_{10}^{L}(S_{00}^{L})^{-1}S_{01}^{L}
\right|
=0,
\label{B.8}
\end{equation}
Let $\widehat\lambda_1\geq\cdots\geq\widehat\lambda_m$ denote the eigenvalues and
$\widehat v_1,\dots,\widehat v_m$ the corresponding eigenvectors. The estimate
$\widehat\gamma$ is obtained by taking the $r_1$ eigenvectors associated with the
largest $r_1$ eigenvalues. 

A convenient statistical normalization is
\begin{equation}
\widetilde\gamma' S_{11}^{L}\widetilde\gamma = I_{r_1}.
\label{B.11}
\end{equation}
Any economically meaningful normalization can then be imposed afterwards. The
normalized matrix is the estimator of the row-side cointegration matrix $\gamma$. 

\paragraph{Step 4. Obtaining the adjustment matrix.}
Given $\widehat\gamma$, the maximum likelihood estimator of the row adjustment matrix
is
\begin{equation}
\widehat\tau
=
S_{01}^{L}\widehat\gamma
\left(\widehat\gamma' S_{11}^{L}\widehat\gamma\right)^{-1}.
\label{B.12}
\end{equation}

\paragraph{Step 5. Recovery of the short-run matrices.}
Once $(\widehat\tau,\widehat\gamma)$ have been obtained, the short-run block
$\Psi_1=[\Gamma_{1,1},\dots,\Gamma_{1,p-1}]$ can be recovered from the pooled
regression of $\widetilde y_{tj}^{L}$ on $\widetilde x_{tj}^{L}$ and
$\widetilde z_{tj}^{L}$ as
\begin{equation}
\widehat\Psi_1
=
\left(M_{yz}^{L}-\widehat\tau\widehat\gamma' M_{xz}^{L}\right)
\left(M_{zz}^{L}\right)^{-1},
\label{B.14}
\end{equation}
where
\begin{equation}
M_{yz}^{L}
=
\frac{1}{N_L}\sum_t\sum_j \widetilde y_{tj}^{L}\widetilde z_{tj}^{L\prime},
\qquad
M_{xz}^{L}
=
\frac{1}{N_L}\sum_t\sum_j \widetilde x_{tj}^{L}\widetilde z_{tj}^{L\prime},
\qquad
M_{zz}^{L}
=
\frac{1}{N_L}\sum_t\sum_j \widetilde z_{tj}^{L}\widetilde z_{tj}^{L\prime}.
\label{B.15}
\end{equation}
The individual matrices $\widehat\Gamma_{1,i}$ are then read from the corresponding
blocks of $\widehat\Psi_1$.

\paragraph{The estimation of Auxiliary\_2.}
The same five-step procedure applies to \eqref{Auxiliary_2} after replacing
$(\tau,\gamma,\Gamma_{1,i})$ with $(\varphi,\theta,\Gamma_{2,i})$ and using the
pooled moment matrices constructed from the whitened right-side system.

\section{Likelihood Ratio Rests in the ECC--MAR Model}

This appendix reports the step-by-step implementation of the three classes of hypothesis
tests discussed in Section~4.2. The procedure is the same as in the standard C-VAR model,
except that in the present matrix setting the relevant sample moment matrices are pooled
across the transformed column-wise observations.

Throughout, the row-side tests are based on the concentrated version of
\eqref{Auxiliary_1},
$$
r_{0,tj}^{L}=\tau\gamma' r_{1,tj}^{L}+u_{tj}^{L},
$$
obtained after partialling out the short-run regressors by Frisch--Waugh. Let
$$
N_L := T_0 q_2
$$
denote the total number of pooled observations and define
$$
S_{ab}^{L}
=
\frac{1}{N_L}\sum_t\sum_j r_{a,tj}^{L} r_{b,tj}^{L\prime},
\qquad a,b\in\{0,1\}.
$$
The column-side tests are constructed analogously from the concentrated version of
\eqref{Auxiliary_2},
$$
r_{0,tj}^{R}=\varphi\theta' r_{1,tj}^{R}+u_{tj}^{R},
$$
with
$$
N_R := T_0 q_1,
\qquad
S_{ab}^{R}
=
\frac{1}{N_R}\sum_t\sum_j r_{a,tj}^{R} r_{b,tj}^{R\prime}.
$$

\subsection*{C.1 Uniform restrictions on the cointegration matrices}

We first consider the row-side hypothesis
$$
H_0^\gamma:\quad \gamma = H_\gamma \widetilde\gamma,
$$
where $H_\gamma$ is an $m\times s_1$ design matrix.

\paragraph{Step 1.}
Estimate the unrestricted row-side model and compute the pooled moment matrices
$S_{00}^{L}$, $S_{11}^{L}$, $S_{01}^{L}$, and $S_{10}^{L}$.

\paragraph{Step 2.}
Solve the unrestricted generalized eigenvalue problem
$$
\left|
\lambda S_{11}^{L}
-
S_{10}^{L}(S_{00}^{L})^{-1}S_{01}^{L}
\right|=0
$$
and order the unrestricted eigenvalues
$$
\widehat\lambda_1\geq \cdots \geq \widehat\lambda_m.
$$

\paragraph{Step 3.}
Under $H_0^\gamma$, solve the restricted eigenvalue problem
$$
\left|
\lambda H_\gamma' S_{11}^{L} H_\gamma
-
H_\gamma' S_{10}^{L}(S_{00}^{L})^{-1}S_{01}^{L} H_\gamma
\right|=0
$$
and order the restricted eigenvalues
$$
\widehat\lambda_1^{\,c}\geq \cdots \geq \widehat\lambda_{s_1}^{\,c}.
$$

\paragraph{Step 4.}
Form the likelihood ratio statistic
$$
LR_\gamma
=
N_L
\left[
\sum_{i=1}^{r_1}\ln(1-\widehat\lambda_i^{\,c})
-
\sum_{i=1}^{r_1}\ln(1-\widehat\lambda_i)
\right].
$$

\paragraph{Step 5.}
Under the null,
$$
LR_\gamma \overset{d}{\longrightarrow} \chi^2\!\left(r_1(m-s_1)\right).
$$

The column-side test for
$$
H_0^\theta:\quad \theta = H_\theta \widetilde\theta
$$
is obtained by replacing $(\gamma,\tau,S_{ab}^{L},N_L,r_1,m,s_1)$ with
$(\theta,\varphi,S_{ab}^{R},N_R,r_2,n,s_2)$. 

\subsection*{C.2 Testing whether a known vector belongs to the cointegration space}

We next consider the row-side hypothesis
$$
H_0^g:\quad \gamma=(g,\widetilde\gamma),
$$
where $g$ is an $m\times r_g$ matrix of known vectors.

\paragraph{Step 1.}
Estimate the unrestricted row-side model and compute the unrestricted pooled eigenvalues
$\widehat\lambda_1,\dots,\widehat\lambda_{r_1}$ from the eigenvalue problem in
Section~C.1.

\paragraph{Step 2.}
Under $H_0^g$, the process
$$
g' r_{1,tj}^{L}
$$
is stationary. Concentrate it out from both $r_{0,tj}^{L}$ and $g_\perp' r_{1,tj}^{L}$
using the Frisch--Waugh theorem, obtaining residuals
$$
r_{0.g,tj}^{L},\qquad r_{1.g,tj}^{L},
$$
where $g_\perp$ denotes the orthogonal complement of $g$.

\paragraph{Step 3.}
Construct the concentrated pooled moment matrices
$$
S_{ab.g}^{L}
=
\frac{1}{N_L}\sum_t\sum_j r_{a.g,tj}^{L} r_{b.g,tj}^{L\prime},
\qquad a,b\in\{0,1\}.
$$

\paragraph{Step 4.}
Solve the auxiliary eigenvalue problem associated with the known vectors:
$$
\left|
\rho S_{00}^{L}
-
S_{01}^{L}g(g'S_{11}^{L}g)^{-1}g'S_{10}^{L}
\right|=0,
$$
and denote the resulting eigenvalues by
$$
\widehat\rho_1\geq\cdots\geq\widehat\rho_{r_g}.
$$

\paragraph{Step 5.}
Solve the restricted eigenvalue problem on the concentrated model:
$$
\left|
\lambda S_{11.g}^{L}
-
S_{10.g}^{L}(S_{00.g}^{L})^{-1}S_{01.g}^{L}
\right|=0,
$$
and denote the resulting eigenvalues by
$$
\widehat\lambda^{\,c}_{1.g}\geq\cdots\geq\widehat\lambda^{\,c}_{r_1-r_g,g}.
$$

\paragraph{Step 6.}
Form the likelihood ratio statistic
$$
LR_g
=
N_L
\left[
\sum_{i=1}^{r_g}\ln(1-\widehat\rho_i)
+
\sum_{i=1}^{r_1-r_g}\ln(1-\widehat\lambda^{\,c}_{i.g})
-
\sum_{i=1}^{r_1}\ln(1-\widehat\lambda_i)
\right].
$$

\paragraph{Step 7.}
Under the null,
$$
LR_g \overset{d}{\longrightarrow} \chi^2\!\left((m-r_1)r_g\right).
$$

The corresponding column-side hypothesis
$$
H_0^t:\quad \theta=(t,\widetilde\theta)
$$
is handled analogously by replacing $(g,g_\perp,S_{ab}^{L},N_L,r_1,m)$ with
$(t,t_\perp,S_{ab}^{R},N_R,r_2,n)$.

\subsection*{C.3 Restrictions on the adjustment matrices}

Finally, consider the row-side hypothesis
$$
H_0^\tau:\quad \tau = H_\tau \widetilde\tau,
$$
where $H_\tau$ is an $m\times s_\tau$ design matrix. Zero-row restrictions are obtained
as a special case and correspond to weak exogeneity.

\paragraph{Step 1.}
Start from the concentrated row-side model
$$
r_{0,tj}^{L}=\tau\gamma' r_{1,tj}^{L}+u_{tj}^{L}.
$$

\paragraph{Step 2.}
Premultiply the system by $H_\tau'$ and $H_{\tau,\perp}'$ and partition the dependent
variable accordingly. Using Gaussian conditioning, write the model for the
$H_\tau' r_{0,tj}^{L}$ block conditional on the $H_{\tau,\perp}' r_{0,tj}^{L}$ block.

\paragraph{Step 3.}
Since the conditioning block is stationary, concentrate it out by Frisch--Waugh from both
the dependent variable and the regressors, obtaining residuals
$$
r_{0.H,tj}^{L},\qquad r_{1.H,tj}^{L}.
$$

\paragraph{Step 4.}
Construct the conditioned pooled moment matrices
$$
S_{ab.H}^{L}
=
\frac{1}{N_L}\sum_t\sum_j r_{a.H,tj}^{L} r_{b.H,tj}^{L\prime},
\qquad a,b\in\{0,1\}.
$$

\paragraph{Step 5.}
Solve the restricted eigenvalue problem
$$
\left|
\lambda S_{11.H}^{L}
-
S_{10.H}^{L}(S_{00.H}^{L})^{-1}S_{01.H}^{L}
\right|=0
$$
and order the eigenvalues
$$
\widehat\lambda_{1,\tau}^{\,c}\geq\cdots\geq\widehat\lambda_{r_1,\tau}^{\,c}.
$$

\paragraph{Step 6.}
Using the unrestricted eigenvalues
$\widehat\lambda_1,\dots,\widehat\lambda_{r_1}$ from Section~C.1, form
$$
LR_\tau
=
N_L
\left[
\sum_{i=1}^{r_1}\ln(1-\widehat\lambda_{i,\tau}^{\,c})
-
\sum_{i=1}^{r_1}\ln(1-\widehat\lambda_i)
\right].
$$

\paragraph{Step 7.}
Under the null,
$$
LR_\tau \overset{d}{\longrightarrow} \chi^2\!\left(r_1(m-s_\tau)\right).
$$
In the important special case of zero-row restrictions, this is the likelihood ratio test for
weak exogeneity of the corresponding row variables.

The test for restrictions on the column-side adjustment matrix
$$
H_0^\varphi:\quad \varphi = H_\varphi \widetilde\varphi
$$
is obtained analogously from the concentrated version of \eqref{Auxiliary_2} by replacing
$(\tau,\gamma,S_{ab}^{L},N_L,r_1,m,s_\tau)$ with
$(\varphi,\theta,S_{ab}^{R},N_R,r_2,n,s_\varphi)$. 

\section{Preliminary Time-Series Analysis for the Empirical Application}

This appendix reports the preliminary univariate and multivariate time-series evidence underlying the empirical application in Section~6. In particular, we document the integration properties of the variables, the lag-order selection, and the determination of the total cointegration rank of the vectorized system.

We begin by testing the order of integration of the six series, both in levels and in first differences, using Dickey--Fuller tests without deterministic terms. Table~\ref{T1} reports the results for the variables in levels, while Table~\ref{T4} reports the results for the first differences.

Overall, the evidence is consistent with all variables being integrated of order one. At the 1\% significance level, the null of a unit root cannot be rejected in levels, whereas it is strongly rejected for the first differences across all series.

\begin{table}[ht]
\centering
\begin{tabular}{@{}lcccccc@{}}
\toprule
\multirow{2}{*}{} & \multicolumn{2}{c}{Industrial Production} & \multicolumn{2}{c}{Expectations on Inflation} \\ \cmidrule(lr){2-3} \cmidrule(lr){4-5}
                         & Statistic         & p-value               & Statistic          & p-value                 \\ \midrule
Belgium                  & 1.0676             & 0.9254               & -1.9596              & 0.0481                   \\
Germany                  & 0.7244            & 0.8704                & -1.2128  & 0.2069                  \\
Austria                  & 2.0068      & 0.9894                & -1.1840 &0.2175                  \\ \bottomrule
\end{tabular}
\caption{Dickey-Fuller Test Results for the variables in levels. The table reports the test statistic and corresponding p-value computed using the test specification without a constant or trend term.}
\label{T1}
\end{table}

\begin{table}[ht]
\centering
\begin{tabular}{@{}lcccccc@{}}
\toprule
\multirow{2}{*}{} & \multicolumn{2}{c}{Industrial Production} & \multicolumn{2}{c}{Expectations on Inflation} \\ \cmidrule(lr){2-3} \cmidrule(lr){4-5}
                         & Statistic         & p-value               & Statistic          & p-value                 \\ \midrule
Belgium                      & -23.8722                & $<$1.0000e-03              & -16.4998      & $<$1.0000e-03                   \\
Germany                          & -17.8777                 & $<$1.0000e-03              & -14.5646     & $<$1.0000e-03                   \\
Austria                    & -18.5007                & $<$1.0000e-03              & -21.3487   & $<$1.0000e-03                          \\ \bottomrule
\end{tabular}
\caption{Dickey-Fuller Test Results for the first differences of the variables. The table reports the test statistic and corresponding p-value computed using the test specification without a constant or trend term.}
\label{T4}
\end{table}

To determine the dynamic order of the system, we fit VAR models in levels with lag orders ranging from 0 to 4 and compare them using the Bayesian Information Criterion (BIC). The results are reported in Table~\ref{T5}.

The BIC is minimized at lag order 1, implying a VAR(1) representation. This corresponds to a VECM with no lagged differences, and therefore to an ECC-MAR specification with no short-run terms.
\begin{table}[ht]
\centering
\begin{tabular}{@{}ccccc@{}}
\toprule
VAR Order (\( k \)) & BIC Value \\ \midrule
0                   & 9.2603e+03   \\
1                   & 6.9924e+03   \\
2                   & 7.0555e+03   \\
3                   & 7.1528e+03  \\
4                   & 7.2581e+03   \\ \bottomrule
\end{tabular}
\caption{BIC values for VAR($k$) models. The table reports the Bayesian Information Criterion (BIC) for vector autoregression models fitted to the levels with lag orders ranging from 0 to 4.}
\label{T5}
\end{table}

Having established that the variables are \(I(1)\), we next determine the total cointegration rank of the vectorized process \(\mathrm{vec}(X_t)\) using the Johansen trace test. The results are reported in Table~\ref{T6}.

The null hypothesis \(r\leq 3\) is rejected, whereas \(r\leq 4\) cannot be rejected. We therefore set the total cointegration rank of the vectorized system equal to \(r=4\).

\begin{table}[ht]
\centering
\begin{tabular}{@{}lccc@{}}
\toprule
\multicolumn{1}{c}{\( H_0 \)} & Statistic & Critical Value & p-Value \\ \midrule
\( r = 0 \)    & 214.1299   & 95.7541        & 0.0010  \\
\( r \leq 1 \) & 128.6881   & 69.8187        & 0.0010  \\
\( r \leq 2 \) & 73.5017    & 47.8564        & 0.0010  \\
\( r \leq 3 \) & 33.9020    & 29.7976        & 0.0160  \\
\( r \leq 4 \) & 12.9282    & 15.4948        & 0.1180  \\
\( r \leq 5 \) & 3.1230     & 3.8415         & 0.0772  \\ \bottomrule
\end{tabular}
\caption{Trace Test Results for Cointegration Rank Determination. The table reports the test statistic, critical value, and p-value for the null hypotheses \(H_0: r \leq i\) (with \(i = 0, 1, \ldots, 5\)).}
\label{T6}
\end{table}

Taken together, the preliminary analysis supports the specification adopted in the main text: the six variables are treated as \(I(1)\), the lag order is set to one, and the total vectorized cointegration rank is fixed at \(r=4\).

\end{document}